\documentclass{article}

\PassOptionsToPackage{numbers, square}{natbib}

\usepackage[final]{neurips_2023}
\usepackage{lmodern}

\usepackage[utf8]{inputenc} \usepackage[T1]{fontenc}    \usepackage{hyperref}       \usepackage{url}            \usepackage{booktabs}       \usepackage{amsfonts}       \usepackage{nicefrac}       \usepackage{microtype}      \usepackage{xcolor}         

\usepackage{mathtools}
\usepackage{amsmath}
\usepackage{amsthm}
\usepackage{amssymb}
\usepackage{comment}
\usepackage{cancel}
\usepackage{tikz}
\usetikzlibrary{backgrounds}
\usepackage{subfigure}
\usepackage{wrapfig}

\usepackage{hyperref}
\newtheoremstyle{exampstyle}
  {\topsep} {\topsep} {} {} {\bfseries} {.} {.5em} {} \theoremstyle{definition}
\newtheorem{definition}{Definition}
\newtheorem{example}{Example}
\newtheorem{theorem}{Theorem}
\newtheorem{lemma}{Lemma}

\newtheorem{proposition}{Proposition}
\newtheorem{corollary}{Corollary}
\newtheorem*{priorwork*}{Prior Work}
\newtheorem*{notation*}{Notation}
\newtheorem{remark}{Remark}
\newcommand{\CartProd}{\bigtimes}

\newcommand{\Val}{\mathrm{Val}}
\newcommand{\Ll}{\mathcal{L}}

\newcommand{\base}{\textnormal{base}}

\newcommand{\Lbase}{\Ll_{\base}}

\newcommand{\Graph}{{\mathcal{G}}}

\newcommand{\RubinOutcomes}{\mathfrak{O}}

 \newcommand{\EndoVars}{{\*V}} \newcommand{\EndoVarsHighLevel}{{\*V}_{\mathrm{H}}} \newcommand{\EndoVarsLowLevel}{{\*V}_{\mathrm{L}}}  \newcommand{\EndoVarsX}{{\*X}}

\newcommand{\ExoVars}{{\*{U}}} 

\newcommand{\EndoSet}{{\*v}}

\newcommand{\HighLevel}{{\mathrm{H}}}
\newcommand{\LowLevel}{{\mathrm{L}}}
\newcommand{\SCM}{{\mathcal{M}}}
\newcommand{\counterfactual}{{\mathrm{cf}}}

\newcommand{\RCM}{{\mathcal{R}}}

\newcommand{\LowLevelModel}{{\mathcal{L}}}
\newcommand{\Ante}{\mathbf{x}}

\newcommand{\RubinUnits}{\mathcal{U}}

\newcommand{\Confounded}{{\dashleftarrow\dashrightarrow}}

\newcommand{\1}{\underline{1}}

\newcommand{\semantics}[1]{[\![\mbox{\em $ #1 $\/}]\!]}

\newcommand{\Unit}{{\*{u}}}

\newcommand{\Proves}{\vdash}

\usepackage{centernot}

\renewcommand{\models}{\vDash}

\usepackage{causality}

\usepackage{xr} 

\makeatletter
\newcommand*{\addFileDependency}[1]{\typeout{(#1)}
  \@addtofilelist{#1}
  \IfFileExists{#1}{}{\typeout{No file #1.}}
}
\makeatother

\title{Comparing Causal Frameworks: Potential Outcomes, Structural Models, Graphs, and Abstractions}

\author{Duligur Ibeling \\
  Department of Computer Science\\
  Stanford University\\
\texttt{duligur@stanford.edu} \\
\And
  Thomas Icard \\
  Department of Philosophy \\
  Stanford University \\
\texttt{icard@stanford.edu} \\
}

\begin{document}

\maketitle

\begin{abstract}
The aim of this paper is to make clear and precise the relationship between the Rubin causal model (RCM) and structural causal model (SCM) frameworks for causal inference. Adopting a neutral logical perspective, and drawing on previous work, we show what is required for an RCM to be \emph{representable} by an SCM. A key result then shows that every RCM---including those that violate algebraic principles implied by the SCM framework---emerges as an abstraction of some representable RCM. Finally, we illustrate the power of this conciliatory perspective by pinpointing an important role for SCM principles in classic applications of RCMs;  conversely, we offer a characterization of the algebraic constraints implied by a graph, helping to substantiate further comparisons between the two frameworks. 
\end{abstract}

Toward the end of the twentieth century several frameworks arose for formalizing  and analyzing problems of causal inference. One of these, associated with Rubin \cite{Rubin1974} and others (see \cite{Imbens2015}), takes the \emph{potential outcome}---formalizing the effect of one variable on another---as a fundamental target of analysis. Causal assumptions in the Rubin causal model (RCM) are naturally encoded as algebraic constraints on potential outcomes, and research in this area has spawned a remarkable body of theoretical and applied work especially in social and biomedical sciences (see \cite{Imbens2020} for a review). 

A second approach, associated with Pearl \cite{Pearl1995} and collaborators (see \cite{Pearl2009} for a textbook treatment; see also \cite{Spirtes2001}), focuses instead on assumptions that can be encoded qualitatively, or more specifically, graphically, arising from a fundamental object known as a structural causal model (SCM). The \emph{do-calculus} is one of the crowning achievements of this approach, and it has been shown derivationally complete with respect to a wide range of canonical estimation and inference problems \citep{ShpitserPearlComplete,BareinboimPearl,lee2019general}.

Both frameworks have enjoyed considerable influence within recent causal machine learning in particular. As just one example, concern in reinforcement learning about the possibility of unobserved confounders---variables impacting both decisions and their outcomes---has generated a number of important advances, some employing tools and concepts from the RCM approach (e.g., \cite{KallusZhou,Namkoong,JongbinGoel}), others grounded in the SCM approach and typically involving graphs (e.g., \cite{bareinboimforney,Tennenholtz2020,Zhang202neurips}). 

Despite the remarkable successes that both of these frameworks have engendered in machine learning and  beyond, there remains substantial controversy over how to understand their relationship. In the literature it has been claimed, on the one hand, that the two are equivalent,  that ``a theorem in one is a theorem in the other'' \citep[p. 98]{Pearl2009}. 
On the other hand,  some authors suggest that the two are only equivalent in a weak sense, one that ``builds the limitations of SCM into the
resulting comparison and likewise filters out aspects of the rival theory that do not
readily translate to SCM'' \citep[p. 443]{Markus2021}.

At issue are two separable questions. The first is one of practical significance. Some argue that graphs give greater ``conceptual clarity'' \citep[p. 103]{Pearl2009} and that SCMs more generally offer a ``a flexible formalism for data-generating models'' that helps ground causal inquiry \citep[p. 514]{bareinboim:etal20}; others argue that work in the RCM framework provides ``transparent definitions of causal effects and encourages the analyst to consider individual-level heterogeneity as a first principle'' \citep[p. 91]{MorganWinship} as well as  ``guidance to researchers and policy makers
for practical implementation'' \citep[p. 1131]{Imbens2020}. While obviously very important, our goal is not to address these disputes about what theoretical primitives are most ``natural'' or ``useful'' for practitioners or applied researchers.  Rather, the aim of the present contribution is to offer a number of new technical results that together shed light on a more basic question, namely, how precisely the RCM and SCM frameworks relate at a theoretical level. For example, are the two merely notational variants, or does one tactictly enforce assumptions that the other does not? 

In this paper we first endeavor, building on previous work, to elucidate the precise sense in which SCMs can be construed as \emph{representations} of RCMs, provided the latter satisfy two key principles known as composition and reversibility \cite{Galles1998,Halpern2000}. Interestingly, such principles (or their logical consequences) have been questioned in the literature (e.g., \cite{colefrangakis}). Our second goal is to help clarify the sense in which they may fail. Drawing from recent literature on causal abstraction (e.g., \cite{Rubinstein2017,beckers20a})---broadened to cover both SCMs and RCMs---we suggest that failure of these principles plausibly results when causally relevant low-level details are elided in favor of more abstract variables. Our Thm.~\ref{thm:abstractionrepresentation} buttresses this intuition, showing that every RCM is a \emph{constructive abstraction} of a representable RCM (hence satisfying composition and reversibility). We furthermore remark on how the well known SUTVA assumptions \cite{Imbens2015} can be understood as conditions on good variable abstractions. 

Our starting point in this work is theoretically neutral, taking for granted only the primitive, ``probability of a counterfactual.'' In the second half of the paper we introduce a framework-neutral formal language for reasoning about such probabilities, which in turn facilities further comparison. 
With respect to this common language, we offer a completeness result for the class of all RCMs (Thm.~\ref{thm:poassumptionscompleteness}) and, drawing on \cite{Halpern2000}, the class of representable RCMs (Cor.~\ref{cor:rep}). These results are illustrated with an example derivation of LATE (see \cite{AngristImbens,Angrist1996}), which also helps illuminate which assumptions are logically required in the derivation. Meanwhile, we offer a partial answer to the well-known open question \cite{TianPearl2002,RichardsonEvans} of how to characterize the algebraic constraints implied by a particular graph (Thm.~\ref{thm:graphicalsystem}), a result that helps bring graphical assumptions into this neutral common language. Finally, we show how an existing result on single-world intervention graphs (SWIGs), a framework drawing from both perspectives, can be construed as a completeness result for the same language (Thm. \ref{thm:swig}).

Taken together, our results are largely conciliatory---in the same spirit as other important conciliatory work in this context; see, e.g., \cite{RichardsonRobins,ShpitserTchetgen,pmlr-v89-malinsky19b}---showing how the two frameworks are productively compatible, while also suggesting distinctive perspectives on problems of causal inference. 

Proofs are deferred to supplementary appendices \ref{s:app:modeling}, \ref{s:app:inference}, which contain additional technical material.

\section{Modeling}\label{sec:modeling}

We first introduce a formalization of the Rubin causal model \citep{Rubin1974,Holland1986,Holland1986a} before then turning to structural causal models \citep{Pearl1995,Pearl2009,bareinboim:etal20}. The relationship between these two is elucidated in \S\ref{section:rep}. 

\subsection{Preliminaries}\label{sec:preliminaries}

Common to both frameworks will be a set $\mathbf{V}$ of \emph{endogenous variables}. Concerning notation:

\begin{notation*}
The signature (or range) of a variable $V$ is denoted $\Val(V)$. 
Where $\*S$ is a set of variables, let $\Val(\*S) = \CartProd_{S \in \*S} \Val(S)$, the product of the family of sets $\Val(S)$ indexed by $S \in \*S$. Elements of $\Val(\*S)$ represent joint valuations of the variables $\*S$.
Given an indexed family of sets $\{S_\beta\}_{\beta \in B}$ and elements $s_\beta \in S_\beta$, let $\{s_\beta\}_\beta$ denote the indexed family of elements whose object associated with the index $\beta$ is $s_\beta$, for all $\beta$.
The symbol $\subset$ indicates any subset (or set inclusion) and does not imply a strict subset (or proper inclusion).
For $B' \subset B$ write $\pi_{B'} : \CartProd_{\beta \in B} S_\beta \to \CartProd_{\beta \in B'} S_{\beta}$ for the \emph{projection map} sending each $\{s_\beta\}_{\beta \in B} \mapsto \{s_{\beta'}\}_{\beta' \in B'}$; abbreviate $\pi_{\beta'} = \pi_{\{\beta'\}}$, for $\beta' \in B$. 
Thus if $\*s \in \Val(\*S)$ is a joint valuation of variables $\*S$ and $S \in \*S$ is a single variable, then $s = \pi_S(\*s) \in \Val(S)$ is a value for $S$. If $\*S' \subset \*S$ then $\pi_{\*S'}(\*s) \in \Val(\*S')$ is a joint valuation of $\*S'$, namely the projection of $\*s$ to $\*S'$.
Upper-case letters like $\mathbf{S}$ conventionally represent those sets of variables that the corresponding lower-case letters $\mathbf{s}$ are valuations of, $\mathbf{s} \in \Val(\*S)$.
\end{notation*}
    
\subsubsection{Rubin Causal Models, Potential Outcomes, Counterfactuals}
The present formalization of the Rubin causal model \citep{Rubin1974,Imbens2015} loosely follows previous presentations; see especially \cite{Holland1986a}. It codifies experimental outcomes across individuals drawn from a distribution. These are \emph{potential outcomes} over a variable set $\*V$, defined as expressions of the form $Y_\Ante$ for an \emph{outcome} $Y \in \*V$ and an \emph{intervention} or \emph{treatment} $\Ante \in \Val(\*X)$ for some $\*X \subset \*V$, and interpreted as the value observed for $Y$ in a controlled experiment where each $X \in \*X$ is held fixed to $\pi_X(\mathbf{x})$.
\begin{definition}
    \label{def:rcm}
A \emph{Rubin causal model} (RCM) is a tuple $\mathcal{R} = \langle \mathcal{U}, \mathbf{V}, \mathfrak{O}, \mathfrak{F}, P \rangle$ where $\mathcal{U}$ is a finite set of \emph{units} or \emph{individuals}, $\mathbf{V}$ is a finite set of \emph{endogenous variables}, $\mathfrak{O}$ is a set of {potential outcomes} over $\*V$, $\mathfrak{F}$ is a set of \emph{potential response} functions, to be defined shortly, and $P : \mathcal{U} \to [0, 1]$ is a probability distribution on $\mathcal{U}$. A potential response for $Y_\Ante \in \mathfrak{O}$ is a function $f_{Y_\Ante}: \mathcal{U} \to \Val(Y)$. For each $Y_\Ante \in \mathfrak{O}$ we require that $\mathfrak{F} = \{f_{Y_\Ante}\}_{Y_\Ante \in \mathfrak{O}}$ contain exactly one such function.\footnote{By generalizing to allow multiple such functions one arrives at a class of models closely related to the \emph{generalized structural equation models} (GSEMs) of Peters and Halpern \cite{peters2021causal}, or the equivalent class of \emph{causal constraint models} (CCMs) introduced by Blom et al. \cite{Blom2019}. Rather than mapping each potential outcome $Y_{\Ante}$ to a value in $\Val(Y)$, GSEMs map each (allowable) intervention $\Ante$ to a (possibly empty) \emph{set} of values for all the variables, that is, to elements of the powerset $\wp{(\Val(\mathbf{V}))}$. RCMs, by contrast, allow that, e.g., $Y_{\Ante}$ may be defined for all $\mathbf{u}$, while $Z_{\Ante}$  is undefined because $Z_{\Ante} \notin \mathfrak{O}$. The two are thus incomparable in expressive power.}
\end{definition}
RCMs are often specified in a tabular form as in, e.g., Fig.~\ref{fig:example:cmreq} below.
We adopt the notation $y_\Ante(i)$ or $Y_\Ante(i) = y$ as a shorthand for $f_{Y_\Ante}(i) = y$:
for $\mathcal{R}$ as in Def.~\ref{def:rcm}, write $\mathcal{R} \vDash y_\Ante(i)$ iff $Y_\Ante \in \mathfrak{O}$ and $f_{Y_\Ante}(i) = y$. This means that in the above controlled experiment, outcome $y \in \Val(Y)$ is observed for individual $i$. Each $Y_\Ante$ can be thought of as a variable with range $\Val(Y_\Ante) = \Val(Y)$.
We call the set $\Val(\mathfrak{O})$ of joint valuations of these variables \emph{counterfactuals}.
A set of potential responses $\mathfrak{F}$ then maps units to counterfactuals, $\mathfrak{F}: \mathcal{U} \to \Val(\mathfrak{O})$, by defining $\mathfrak{F}(i) = \{f_{Y_\Ante}(i)\}_{Y_\Ante \in \mathfrak{O}}$, and:

\begin{definition}
Where $\mathcal{R}$ is as in Def.~\ref{def:rcm}, the \emph{counterfactual distribution} $P_{\mathrm{cf}}^\mathcal{R} : \Val(\mathfrak{O}) \to [0, 1]$ induced by $\mathcal{R}$ is the pushforward\footnote{That is, $P_{\mathrm{cf}}^\mathcal{R}(\*o) = P(\mathfrak{F}^{-1}(\{\*o\}))$ for any $\*o \in \Val(\mathfrak{O})$.} $\mathfrak{F}_*(P)$ of $P : \mathcal{U} \to [0, 1]$ under $\mathfrak{F}: \mathcal{U} \to \Val(\mathfrak{O})$. \label{def:counterfactualrcm}
\end{definition}
The reason we call $P_{\mathrm{cf}}^\mathcal{R}$ a counterfactual distribution (and $\Val(\mathfrak{O})$ counterfactuals) is because such joint probabilities over multiple potential outcomes appear in the usual ratio definition of counterfactual probabilities. For instance, $P(y_x|y'_{x'}) = P(y_x,y'_{x'})/P(y'_{x'})$ gives the probability that a unit who was withheld treatment and did not recover would have recovered if assigned treatment. But $P_{\mathrm{cf}}^\mathcal{R}$ also answers via marginalization all questions (whenever defined by $\mathcal{R}$) about ``interventional'' probabilities like $P(y_\mathbf{x})$, as well as purely ``observational'' probabilities such as $P(\mathbf{x})$; see, e.g., \cite{bareinboim:etal20}. 

Some authors submit that ``probability will
 mean nothing more nor less than a proportion of units'' \citep[p. 945]{Holland1986}, thereby assuming a uniform distribution on a finite population $\mathcal{U}$ (cf. also \cite{Angrist1996}). Of course, in the infinite population size limit we recover all RCMs as in Def.~\ref{def:rcm} (see Prop.~\ref{prop:app:limitofunits}).

While practitioners do not typically consider potential outcomes $Y_\Ante$ when $Y \in \mathbf{X}$, instead maintaining a strict dichotomy between cause and effect variables (e.g., \cite{Holland1986,Holland1986a}), it is natural to impose the following requirement (known as \emph{effectiveness}) whenever such potential outcomes are defined. An intervention is always assumed to be a \emph{successful} intervention: whenever $Y \in \mathbf{X}$,
\begin{eqnarray}\label{eq:poeffectiveness}
\text{Effectiveness}. && Y_\Ante(u) = \pi_Y(\Ante)
\end{eqnarray}
for every $u \in \mathcal{U}$.
In fact, practice in the RCM framework reflects this same assumption, in the sense that violations of it are taken to signify a poor choice of variables. As a classic example, the possibility of \emph{non-compliance} in an experimental trial motivates the introduction of instrumental variables, and specifically a separation between, e.g., \emph{treatment} and \emph{assignment to treatment} (cf. Ex.~\ref{eg:late}). Crucially, we recover effectiveness with respect to the latter. 
We will assume any RCM to meet \eqref{eq:poeffectiveness} unless otherwise specified; let $\mathfrak{R}_{\mathrm{eff}}$ be the class of such RCMs.\footnote{Not only does this assumption reflect practice, but it is also without loss with regard to comparing RCMs and SCMs, as the latter also satisfy effectiveness: see footnote \ref{footnote:reffexplain}.}

\subsubsection{Structural Causal Models}
An important feature of RCMs is that potential outcomes are cleanly separated from assignment mechanisms \citep{Imbens2015}. A different starting point is to assume that potential outcomes and their algebraic behavior are rather \emph{derived} from an underlying formal representation of causal structure. These putatively ``deeper mathematical objects'' \citep[p. 102]{Pearl2009} involve concrete functional dependencies, and an operation of function replacement known as \emph{intervention}:
\begin{definition}\label{def:scm}
    A \emph{structural causal model} (SCM) is a tuple $\mathcal{M} = \langle \mathbf{U}, \mathbf{V}, \mathcal{F},  P(\mathbf{U}) \rangle$ where $\mathbf{U}$ is a finite set of \emph{exogenous variables}, $\mathbf{V}$ is a finite set of \emph{endogenous variables}, $\mathcal{F}$ is a family of \emph{structural functions}, to be defined shortly, and $P : \Val(\*U) \to [0, 1]$ is a probability distribution on (joint valuations of) $\*U$.
    
    A structural function for $V \in \*V$ is a function
    $f_V: \Val(\mathbf{U}_V \cup \mathbf{Pa}_V) \to \Val(V)$, where $\*U_V \subset \*U$, $\mathbf{Pa}_V \subset \*V \setminus \{ V\}$.
    For every $V \in \*V$ we require that $\mathcal{F} = \{f_V\}_V$ have exactly one such function; the entire collection $\mathcal{F}$ thus forms an exogenous-to-endogenous mapping.
\end{definition}
Interventions come as a derived notion, replacing structural functions with constant functions \citep{meekglymour94,Pearl2009}: \begin{definition}[Intervention on SCMs]\label{def:intervention}
Let $\*x \in \Val(\*X)$ for some $\*X \subset \*V$ be an intervention and $\mathcal{M}$ be the SCM of Def.~\ref{def:scm}.
Then define a \emph{manipulated model} $\mathcal{M}_{\*x} = \langle \mathbf{U}, \mathbf{V}, \{f'_V\}_V, P(\*U) \rangle$ where each $f'_V : \Val(\mathbf{U}'_V \cup \mathbf{Pa}'_V) \to \Val(V)$.
If $V \notin \*X$ define $\mathbf{U}'_V = \*U_V$, $\mathbf{Pa}'_V = \mathbf{Pa}_V$, and $f'_V = f_V$.
If $V \in \*X$ define $\mathbf{U}'_V = \mathbf{Pa}'_V = \varnothing$ and $f'_V$ as a constant function mapping to $\pi_V(\*x)$.
\end{definition}
Letting $\mathcal{M}$ be the SCM of Def.~\ref{def:scm}, for $\*v \in \Val(\*V)$ and $\*u \in \Val(\*U)$ write $\mathcal{M}, \*u \models \*v$ if we have $f_V\big(\pi_{\mathbf{U}_V \cup \mathbf{Pa}_V}(\*v)\big) = \pi_V(\*v)$ for every $V$.
Let $\mathfrak{M}_{\mathrm{uniq}}$ be the class of all SCMs $\mathcal{M}$ such that,
for any $\*u$ and intervention $\*x$ there is a unique \emph{solution} $\*v$ such that $\mathcal{M}_{\*x}, \*u \models \*v$. In this case we define the potential outcome $Y_{\*x}(\*u)$ as $\pi_Y(\*v)$.
Thus any $\mathcal{M} \in \mathfrak{M}_{\mathrm{uniq}}$ defines a potential outcome for \emph{every} $Y_\Ante$,
giving a natural function $p^{\mathcal{M}}: \Val(\mathbf{U}) \to \Val(\{Y_\Ante\}_{\text{all } Y, \Ante})$ via these outcomes, and:
\begin{definition} \label{def:counterfactualscm}
    The counterfactual distribution $P_{\mathrm{cf}}^{\mathcal{M}} : \Val(\{Y_\Ante\}_{\text{all }Y,\Ante}) \to [0, 1]$ induced by $\mathcal{M} \in \mathfrak{M}_{\mathrm{uniq}}$ is the pushforward $p^{\mathcal{M}}_*(P)$ of $P : \Val(\mathbf{U}) \to [0, 1]$ under $p^{\mathcal{M}} :\Val(\mathbf{U}) \to \Val(\{Y_\Ante\}_{\text{all } Y, \Ante})$.
\end{definition} Thus, SCMs in $\mathfrak{M}_{\mathrm{uniq}}$ canonically define counterfactual distributions for all possible potential outcomes via manipulation of functional dependencies.
Importantly, Def. \ref{def:counterfactualscm} provides a bridge to RCMs, as both produce counterfactual distributions (recall Def. \ref{def:counterfactualrcm}). As long as the counterfactual probabilities are assumed to mean the same thing---i.e., as long as they highlight the same targets for empirical and theoretical investigation---we can then compare the ranges of assumptions and inference patterns that each framework can encode about them. 
We thus assume that all our SCMs belong to this class $\mathfrak{M}_{\mathrm{uniq}}$.

\subsubsection{Representation of RCMs by SCMs} \label{section:rep}
A contentious methodological question is whether all (endogenous) variables should be potential candidates for intervention. Following the literature we have supposed that SCMs allow all possible interventions (though this assumption is not universal; see, e.g., \cite{Rubinstein2017,beckers20a}). For RCMs it is generally assumed that there can be ``no causation without manipulation'' \citep{Holland1986,Imbens2015}, and thus that only some interventions should be allowed. While methodologically important, this is theoretically inessential. We can construe SCMs as possible representations of RCMs in the following sense:
\begin{definition}
Let $\mathcal{R} \in \mathfrak{R}_{\mathrm{eff}}$ and $\mathcal{M} \in \mathfrak{M}_{\mathrm{uniq}}$.
We say that $\mathcal{M}$ \emph{represents} $\mathcal{R}$ if its counterfactual distribution $P_\counterfactual^\mathcal{M}$ marginalizes down to $P_\counterfactual^\mathcal{R}$ on the potential outcomes defined by $\mathcal{R}$ (the set $\mathfrak{O}$ in Def.~\ref{def:rcm}). We say $\mathcal{R}$ is \emph{representable} if it is represented by at least some $\mathcal{M} \in \mathfrak{M}_{\mathrm{uniq}}$.\footnote{With regard to representability, our assumption that $\mathcal{R} \in \mathfrak{R}_{\mathrm{eff}}$ is without loss since Def.~\ref{def:intervention} easily implies that the potential outcomes induced by any SCM must satisfy effectiveness \eqref{eq:poeffectiveness}.\label{footnote:reffexplain}}
\end{definition}
Thus $\mathcal{M}$ represents $\mathcal{R}$ if they are counterfactually equivalent with respect to the outcomes defined by $\mathcal{R}$. Toward a characterization of representability, consider two properties of an RCM \cite{Galles1998}:
\begin{definition}
    The following Boolean formulas encode assumptions about potential outcomes:
\begin{eqnarray}
    \text{Composition.} && Y_{\*w}(u) = y \land Z_{\*w}(u) = z \rightarrow Z_{\*w y}(u) = z\label{eq:pocomposition}\\
    \text{Reversibility.} && Y_{\*w z}(u) = y \land Z_{\*w y}(u) = z \rightarrow Y_{\*w}(u) = y.\label{eq:poreversibility}
\end{eqnarray}
Say $\mathcal{R} \in \mathfrak{R}_{\mathrm{eff}}$ satisfies composition and reversibility, respectively, when the respective statements hold for every unit $u$ of $\mathcal{R}$, whenever all the appropriate potential outcomes are defined.\end{definition}

We understand lower-case values like $y$, $z$, $\*w$, when not bound as dummy indices or otherwise, to be schematic variables carrying tacit universal quantifiers. Thus \eqref{eq:pocomposition}, \eqref{eq:poreversibility} must hold for all possible $y \in \Val(Y)$, $\*w \in \Val(\*W)$, $z \in \Val(Z)$. This same usage is repeated, e.g., in
\eqref{eqn:cfsep}.

While reversibility seems not to have arisen in the potential outcomes literature, instances of composition have appeared explicitly (e.g., \citealt{Holland1986a}, p. 968) and have been used implicitly in concrete derivations (see Ex.~\ref{eg:late} below).
Note that the well-known principle of \emph{consistency} \cite{colefrangakis,Pearl2009} is merely the instance of composition for $\*W = \varnothing$.
For $\mathcal{R}, \mathcal{R}' \in \mathfrak{R}_{\mathrm{eff}}$ that share the same units $\mathcal{U}$ and endogenous variables $\*V$ but have respective potential outcome sets $\mathfrak{O}, \mathfrak{O}'$ and potential response sets $\mathfrak{F}, \mathfrak{F}'$, if $\mathfrak{O} \subset \mathfrak{O}'$ and $\mathfrak{F} \subset \mathfrak{F}'$ we say that $\mathcal{R}'$ \emph{extends} or is an extension of $\mathcal{R}$ and $\mathcal{R}$ is a \emph{submodel} of $\mathcal{R}'$. Call $\mathcal{R}$ \emph{full} if it has no proper extension.
Then: \begin{proposition}[SCM Representation]\label{prop:scmrepresentation}
RCM $\mathcal{R}$ is representable iff $\mathcal{R}$ extends to some full $\mathcal{R}'$ that satisfies composition and reversibility.
\end{proposition}
Note that for an RCM $\mathcal{R}$ to be representable it is necessary (though not sufficient, in light of the models presented in Fig.~\ref{fig:example:cmreq} below) that $\mathcal{R}$ itself witness no composition or reversibility failures.
Prop.~\ref{prop:scmrepresentation} thus clarifies a sense in which RCMs are more general than SCMs, not just by allowing only a subset of allowable interventions, but also by imposing fewer requirements on how potential outcomes relate to one another. However, assuming composition, reversibility, and fullness, the two define the same classes of counterfactual distributions, despite the superficial differences in their definitions. In that case the two are, e.g., equivalent for interpreting the probabilistic logical language of \S{}\ref{sec:inference}.  We submit that some version of this result makes sense of statements in the literature, e.g., from Pearl \cite{Pearl2009}, that the twain are essentially equivalent frameworks from a theoretical perspective.

\subsection{Causal Abstraction} \label{section:abstraction}
The goal of this section is to clarify the source of putative failures of  principles like composition. We suggest that it is helpful to view these issues through the lens of \emph{causal abstraction} (the definitions in this section are adapted from \cite{Rubinstein2017,beckers20a}). Abstraction has mostly been studied in the context of SCMs; our definitions apply equally to SCMs and RCMs via counterfactual distributions. 

In causal abstraction, one has a set $\EndoVarsLowLevel$ of low-level (or concrete, or micro-) variables representing a fine-grained description and a set $\EndoVarsHighLevel$ of high-level (or abstract, or macro-) variables representing a coarser-grained description of the same scenario.
The correspondence between the two descriptions is given by a partial \emph{translation} map $\tau: \Val(\EndoVarsLowLevel) \to \Val(\EndoVarsHighLevel)$. Translations extend canonically to maps of partial valuations (e.g., interventions) $\tau: \bigcup_{\*X \subset \EndoVarsLowLevel} \Val(\*X) \to \bigcup_{\*X \subset \EndoVarsHighLevel} \Val(\*X)$ by setting $\tau(\*x_{\mathrm{L}}) = \*x_{\mathrm{H}} \text{ iff }
    \tau\big(\pi^{-1}_{\*X_\LowLevel}(\*x_{\mathrm{L}})\big) = \pi^{-1}_{\*X_\HighLevel}(\*x_{\mathrm{H}})$.

We overload $\tau$ once more so as to cover counterfactuals, defining as follows yet another partial $\tau : \Val(\mathfrak{O}_\LowLevel) \to \Val(\mathfrak{O}_\HighLevel)$ for any sets $\mathfrak{O}_\LowLevel$, $\mathfrak{O}_\HighLevel$ of potential outcomes over $\EndoVarsLowLevel$ and $\EndoVarsHighLevel$ respectively.
Index an element of $\Val(\mathfrak{O}_\LowLevel)$ as $\{(\mathbf{y}^i_\LowLevel)_{\*x^i_\LowLevel}\}_{1 \le i \le m}$, where $\*x^i_\LowLevel \neq \*x^j_\LowLevel$ for any $i \neq j$ and $\mathbf{y}^i_\LowLevel \in \Val(\{ Y \in \*V_\LowLevel : Y_{\*x^i_\LowLevel} \in \mathfrak{O}_\LowLevel \})$ for each $i$, and an element of $\Val(\mathfrak{O}_\HighLevel)$ likewise as $\{(\mathbf{y}^j_\HighLevel)_{\*x^j_\HighLevel}\}_{1 \le j \le n}$.
Define $\tau\big(\{(\mathbf{y}^i_\LowLevel)_{\*x^i_\LowLevel}\}_{1 \le i \le m}\big) = \{(\mathbf{y}^j_\HighLevel)_{\*x^j_\HighLevel}\}_{1 \le j \le n}$
if
$\tau(\{\Ante^i_\LowLevel: 1 \le i \le m \}) = \{\Ante^j_\HighLevel: 1 \le j \le n\}$
and $\tau(\mathbf{y}^i_\LowLevel) = \mathbf{y}^j_\HighLevel$ for any pair $
\Ante^i_\LowLevel, \Ante^j_\HighLevel$ where $\tau(\Ante^i_\LowLevel) = \Ante^j_\HighLevel$.
\begin{definition}\label{def:rcmabstraction}
With counterfactual translation in hand, we define an abstraction relation between probabilistic causal models. The model $\mathcal{H}$ abstracts $\LowLevelModel$ over the aligned variables (written $\mathcal{H}\prec_{\tau}\mathcal{L}$) if the translation $\tau$ pushes the latter's counterfactual distribution to the former's, that is, $P^{\mathcal{H}}_\counterfactual = \tau_*(P^{\mathcal{L}}_\counterfactual)$.
\end{definition}

A stricter and typically more useful notion is that of \emph{constructive} abstraction (e.g., \cite{beckers20a}). These arise from translations that can be generated variable-wise, and thus correspond to a coherent ``clustering'' of variables:
\begin{definition}\label{def:constructiveabs}
Translation $\tau: \Val(\EndoVarsLowLevel) \to \Val(\EndoVarsHighLevel)$ is constructive if
there is a partition $\Pi$ of ${\*V}_{\mathrm{L}}$ with non-overlapping cells $\{\Pi_{V}\}_{V \in {\*V}_{\mathrm{H}} \cup \{\bot\}}$, each $\Pi_V \subset {\*V}_{\mathrm{L}}$, where $\Pi_V$ is non-empty for all $V \neq \bot$, and a collection $\{\tau_V\}_{V \in \EndoVarsHighLevel}$ each of which is a partial surjective map $\tau_V: \Val(\Pi_V) \to \Val(V)$, such that $\tau(\EndoSet_\LowLevel) = \big\{ \tau_V\big(\pi_{\Pi_V}(\EndoSet_\LowLevel)\big)\big\}_{V \in \EndoVarsHighLevel}$ for any $\EndoSet_\LowLevel \in \Val(\EndoVarsLowLevel)$.
\end{definition}

A simple abstraction, ubiquitous in the literature (see, e.g., \cite[\S{}1.6.2]{Imbens2015} and \cite{colefrangakis}), is that of variable treatment levels.
Here a higher-level value corresponds to some collection of lower-level specifications that might
represent the potency or dosage of the drug administered, the time of administration, etc.: for example, a distinction of whether one took 300, 400, 500, or 600 mg of aspirin is made at the low level, but at the high level, there is only the binary distinction between having taken aspirin and not.
Formally, a treatment variable $T$ is only binary with values $\mathrm{c}, \mathrm{tr}$ (control, treatment resp.) at the high level but takes on many values $\mathrm{c}, \mathrm{tr}^1, \dots, \mathrm{tr}^n$ at the low level. The abstraction is made by omitting the fine-grained details; symbolically, one forms a new model by eliding the superscripts, collapsing all $\mathrm{tr}^i$ into $\mathrm{tr}$. So long as for any outcomes we have $Y_{\mathrm{tr}^i}(u) = Y_{\mathrm{tr}^j}(u)$, the model thus formed will be a constructive probabilistic abstraction of the low-level model.

The next result provides some useful properties of constructive abstraction.
\begin{proposition}\label{prop:abstractionproperties}
Suppose $\mathcal{H}\prec_{\tau}\mathcal{L}$ with $\tau$ constructive.
Then $\mathcal{H}$ is effective if $\mathcal{L}$ is effective.
Also, for any submodel $\mathcal{H}'$ of $\mathcal{H}$ there is a submodel $\mathcal{L}'$ of $\mathcal{L}$ such that $\mathcal{H}' \prec_{\tau} \mathcal{L}'$. \end{proposition}
Thus our general class of effective RCMs closes under constructive translation.
The next example shows that this is not the case for the narrower class of representable models.

\begin{example}\label{ex:abstraction}
Let $X, Y, X', Y'$ be variables with
$\Val(X) = \{0, 1, 2\}$ and $\Val(X') = \Val(Y') = \Val(Y) = \{0, 1\}$.
Consider the RCM $\RCM_\LowLevel$ defined over $\EndoVars_\LowLevel = \{ X, Y\}$ as a conjunction of POs:
\begin{align}
    X = 1 \land Y = 1 \land Y_{X = 2} = 0 \land X_{Y = 0} = 1 \label{eq:abstractionexample:lowlevel}
\end{align}
for a single unit (suppressed above for clarity).
Consider a second RCM $\RCM_\HighLevel$ over $\EndoVars_\HighLevel = \{X', Y'\}$: \begin{align}
    X' = 1 \land Y' = 1 \land Y'_{X' = 1} = 0 \land X'_{Y' = 0} = 1. \label{eq:abstractionexample:highlevel}
\end{align}
Note that $\RCM_\HighLevel \prec_\tau \RCM_\LowLevel$ where $\tau$ is a constructive abstraction with $\Pi_{X'} = \{X\}$, $\Pi_{Y'} = \{Y\}$ and $\tau(X = 0) = 0$, $\tau(X = 1) = \tau(X = 2) = 1$, $\tau(Y = y) = y$.
Now $\RCM_\HighLevel$ violates both composition and reversibility, while $\RCM_\LowLevel$ is representable.
\end{example}

A second observation is that the analogue of the claim about submodels in Prop.~\ref{prop:abstractionproperties} does not hold for extensions:
\begin{example}
Consider enlarging \eqref{eq:abstractionexample:lowlevel} with the potential outcome $Y_{X = 1} = 1$.
Then there is no high-level abstraction under $\tau$ that defines the outcome $Y'_{X' = 1}$, since $Y_{X = 2} = 0 \land Y_{X = 1} = 1$ translates to $Y'_{X' = 1} = 0 \land Y'_{X' = 1} = 1$.
\end{example}

The main result of this section is that the phenomenon exhibited by Ex.~\ref{ex:abstraction} accounts for all representability failures:
\begin{theorem}[Abstract Representation]\label{thm:abstractionrepresentation}
    Let $\mathcal{R}$ be an RCM. Then there is a representable $\RCM_\LowLevel$ and constructive translation $\tau$ such that $\mathcal{R} \prec_\tau \mathcal{R}_\LowLevel$.
\end{theorem}

It is worth remarking on the connection between a well-known twofold condition called the Stable Unit Treatment Value Assumption (SUTVA \cite[\S{}1.6]{Imbens2015}) and causal abstraction.
The first part of SUTVA is the assumption that ``the potential outcomes for any unit do not vary with the treatments assigned to other units''; this is already presumed by our definition of causal model, which does not admit interventions on multiple units (however, see Rmk.~\ref{rmk:app:sutva} for a way to model this condition within our framework as an application of abstraction).
The second part is that ``for each unit, there are no different forms or versions of each treatment level, which lead to different potential outcomes.'' Note that this assumption can be seen as guaranteeing the viability of the variable treatment levels abstraction, as it is simply a restatement of the condition we already identified---that the outcomes $Y_{\mathrm{tr}^i}(u)$ for any unit $u$ and treatment level $\mathrm{tr}^i$ must all agree.\footnote{Imbens and Rubin \cite{Imbens2015} also mention ways of fulfilling this condition requiring changes to the causal model definition. In the supplement (Rmk.~\ref{rmk:app:sutva2}) we show these alternative models can be represented within our framework.}

\section{Inference}\label{sec:inference}
The raison d'\^{e}tre of a causal inference framework is to provide a language for encoding causal assumptions and showing when and why conclusions follow from available data and appropriate assumptions. In this section, to provide further granularity on the comparison between RCMs and SCMs, we introduce a neutral formal language that is naturally interpreted relative to both of these models. The language systematizes reasoning about the probabilities of counterfactuals. 
Fixing a set $\mathfrak{O}$ of potential outcome pairs, we define a formal language $\mathcal{L}$ in two stages:

\begin{definition} The \emph{base language} $\mathcal{L}_{\textnormal{base}}$ is given by all Boolean combinations of statements $Y_{\Ante}=y$, alternatively written $y_{\Ante}$, for all $Y_\Ante \in \mathfrak{O}$, $y \in \Val(Y)$. Meanwhile, $\mathcal{L}$ is defined as the set of Boolean combinations of inequalities $\mathbf{t}_1 \geqslant \mathbf{t}_2$, where $\mathbf{t}_1,\mathbf{t}_2$ are generated as sums, products, and additive inverses of probability terms $\mathbf{P}(\varepsilon)$, where $\varepsilon \in \mathcal{L}_{\textnormal{base}}$. \end{definition}
The language $\mathcal{L}$ is the most expressive in a sequence of three languages introduced in \cite{ibelingicard2020,bareinboim:etal20} to formalize the ``causal hierarchy'' \citep{Pearl2009}. By restricting probability terms to purely ``observational'' or ``interventional'' quantities, it is possible to study the inferential limitations of data and assumptions at lower levels of this hierarchy. For present purposes, $\mathcal{L}$ naturally encodes prominent reasoning patterns in RCMs and in SCMs. Its semantics are straightforward in any $\SCM$ or $\RCM$ that includes all outcomes $\mathfrak{O}$: we generate a mapping of each polynomial term $\mathbf{t} \mapsto \semantics{\mathbf{t}} \in \mathbb{R}$ recursively with the crucial case being to map $\mathbf{P}(\varepsilon)$ to the probability calculable by marginalization of $p_\counterfactual^\RCM$ or $p_\counterfactual^\SCM$,
and then evaluate the atom $\mathbf{t}_1 \geqslant \mathbf{t}_2$ true iff $\semantics{\mathbf{t}_1} \ge \semantics{\mathbf{t}_2}$, recursing to define a semantics for all of $\Ll$. Over the class of all (recursive, possibly infinite) SCMs, $\mathcal{L}$ has been axiomatized \cite{ibelingicard2020} by a set of principles called $\textsf{AX}_3$, and the complexity of its satisfiability problem has been shown complete for the class $\exists \mathbb{R}$ \cite{mosse}.
The class of simple probability distributions over the atoms of $\Lbase$ is axiomatized by principles known as $\textsf{AX}_1$ \cite{ibelingicard2020}, which we will abbreviate $\textsf{AX}$.

\subsection{Potential Outcomes Assumptions}\label{sec:poassumptions:inference}
Reasoning about potential outcomes is often encoded in what we call the base language $\mathcal{L}_{\textnormal{base}}$, augmented with (typically implicit universal) quantifiers over units. For instance, the well known \emph{monotonicity} (or ``no defiers'' who do the opposite of their prescription) assumption \citep{AngristImbens,Imbens2015} says \begin{equation}\label{eq:pomonotonicity}
    \forall u.  X_{z^-}(u)=1 \rightarrow X_{z^+}(u)=1,
\end{equation} where $X$ and $Z$ are binary variables respectively meaning the treatment (actually taken) and the treatment prescribed, with $z^+,z^-$ abbreviating $Z=1, Z=0$ respectively.
We will use the same abbreviation for other binary variables, so that the above condition can be written succinctly as $x^+_{z^+} \rightarrow x^+_{z^-}$.
We also adopt this interpretation of $X, Z$ for the rest of \S{}\ref{sec:poassumptions:inference}.
We now explain how this and other causal assumptions in the potential outcomes framework can be encoded in $\mathcal{L}$:
\begin{definition}\label{def:equationaltranslation}
    Let $\xi$ be a well-formed, closed predicate formula in prenex normal form with a single quantifier over a variable $\{u\}$ and matrix in $\Lbase$; the $u$ can alternately be included in the atoms, e.g., by writing $Y_\Ante(u) = y$.
    Define its encoding $\mathrm{T}(\xi) \in \mathcal{L}$ as follows:
    \begin{equation*}
        \mathrm{T}(\xi) = \begin{cases}
\mathbf{P}(\lnot \mathrm{T}(\zeta)) = 0, & \xi = \forall u . \zeta, \text{ where } \zeta \in \mathcal{L}_{\textnormal{base}} \\
            \mathbf{P}(\mathrm{T}(\zeta)) > 0, & \xi = \exists u . \zeta, \text{ where } \zeta \in \mathcal{L}_{\textnormal{base}} \\
\end{cases}.
    \end{equation*} \end{definition}
Note that $\zeta${ is quantifier-free} in both cases.
Thus, e.g., the encoding of \eqref{eq:pomonotonicity} is $\mathbf{P}\big[\lnot(x^+_{z^-} \rightarrow x^+_{z^+})\big] = 0$.

Where $S$ is a set of $\Lbase$ assumptions let $\mathfrak{R}(S)$ be the class of RCMs whose potential outcomes obey every assumption in $S$, thus obeying $\forall u. \sigma$ where $u$ ranges over units for any $\sigma \in S$.
Also let $\mathrm{T}(S) = \{ \mathrm{T}(\forall u.\sigma) \}_{\sigma \in S}$ be the encoding of $S$ via Def.~\ref{def:equationaltranslation}.
Then we have the following: \begin{theorem}\label{thm:poassumptionscompleteness}
$\textsf{AX} + \mathrm{T}(S)$ is sound and complete for $\mathfrak{R}(S)$.
\end{theorem}
\begin{corollary}\label{cor:reffax}
    $\textsf{RCM} = \textsf{AX}+ \mathrm{T}\eqref{eq:poeffectiveness}$ is sound and complete for $\mathfrak{R}_{\mathrm{eff}}$.
\end{corollary}
One consequence of this completeness result is that purely propositional and predicate logic reasoning about potential outcomes can be interweaved with probabilistic reasoning, as in Ex.~\ref{eg:late} below.
Another consequence is a complete axiomatization of SCMs (which can be seen as a probabilistic lift of \cite{Halpern2000}):
\begin{corollary}
Let $\textsf{C}$, $\textsf{Rev}$ be universal statements of \eqref{eq:pocomposition} and \eqref{eq:poreversibility} respectively.
Then $\textsf{SCM} = \textsf{RCM} + \mathrm{T}(\textsf{C}) + \mathrm{T}(\textsf{Rev})$ is sound and complete for $\mathfrak{M}_{\mathrm{uniq}}$ (where every outcome is included in $\Lbase$). \label{cor:rep}
\end{corollary}

\begin{example}\label{eg:late}
A seminal result from \cite{AngristImbens,Angrist1996} is that it is possible to estimate the Average Treatment Effect among the population of units who \emph{comply} with their treatment assignment, a quantity known as the \emph{Local} Average Treatment Effect (LATE): $\mathbf{E}(Y_{x^+} - Y_{x^-} | x^+_{z^+} \land x^-_{z^-})$, with $Y$ the outcome, which we assume binary purely for simplicity, and without loss of generality. Thm.~\ref{thm:poassumptionscompleteness} implies that this can be verified in our calculus, by appeal to two key assumptions \cite{AngristImbens,Angrist1996}:  monotonicity \eqref{eq:pomonotonicity} and 
\begin{eqnarray}
\text{Exclusion restriction (ER).} && \forall u. y_{z^-, x}(u) \leftrightarrow y_{z^+, x}(u).\label{eq:po:er}
\end{eqnarray} The original discovery was that these principles guarantee that $\text{LATE} = \text{ITT}_1/\text{ITT}_2$,  where the latter are
the average ``causal effects of assignment on treatment received and on the outcome of interest'' \cite{Angrist1996}, or symbolically:
\begin{align*}
    \text{ITT}_2 &= \mathbf{E}(X_{z^+} - X_{z^-}) = \mathbf{P}(x^+_{z^+} \land x^-_{z^-}) - \mathbf{P}(x^-_{z^+} \land x^+_{z^-})  \\
\text{ITT}_1 &= \mathbf{E}\big(Y_{z^+, X_{z^+}} - Y_{z^-, X_{z^-}}\big) = \mathbf{P}\big(y^+_{z^+, X_{z^+}} \land y^-_{z^-, X_{z^-}}\big) - \mathbf{P}\big(y^-_{z^+, X_{z^+}} \land y^+_{z^-, X_{z^-}}\big). \end{align*}
where in $\text{ITT}_1$, interventions like $X_{z}$ set $X$ at the unit level to the value that it would take under the intervention setting $Z$ to $z$; thus, e.g., we have $\mathbf{P}(y_{z, X_{z}}) = \mathbf{P}(y_{z, x^+} \land x^+_{z}) + \mathbf{P}(y_{z, x^-} \land x^-_{z})$.
Crucially, these two quantities can be estimated, e.g., through randomized experiments \cite[Ch. 23]{Imbens2015}.

However, over our most general class $\mathfrak{R}_{\mathrm{eff}}$ of RCMs, these two assumptions are not in fact sufficient to identify $\text{LATE}$.  Fig. \ref{fig:example:cmreq} illustrates a family of RCMs that satisfy \eqref{eq:pomonotonicity} and \eqref{eq:po:er}, but disagree on LATE. An additional principle, which Angrist et al. \cite{Angrist1996} offer as a matter of notation, we dub:
\begin{eqnarray}
\text{Outcome decomposition.} && \forall u. y_x(u) \leftrightarrow y_{z^+,x}(u).\label{eq:po:od}
\end{eqnarray}
It can then be shown that, taken together, \eqref{eq:pomonotonicity},  \eqref{eq:po:er}, and \eqref{eq:po:od} do indeed logically entail $\text{LATE} = \text{ITT}_1/\text{ITT}_2$; see Prop.~\ref{prop:appx:late} in the technical appendix for the derivation. 

There has been much discussion of monotonicity and exclusion restrictions (which are closely related to graphical assumptions; see \S\ref{section:graph} below), but what might justify outcome decomposition \eqref{eq:po:od}? One intuition might be that it somehow follows from the exclusion restriction \eqref{eq:po:er}: if the effect of $X$ on $Y$ is the same no matter the value of $Z$, then it would seem that omitting $z^+$ in the intervention should have no impact on the effect of $X$ on $Y$. Of course, the example in Fig. \ref{fig:example:cmreq} shows that this is too quick. 

It turns out that \eqref{eq:po:od} does follow from \eqref{eq:po:er} if we restrict attention to \emph{representable} RCMs. In fact, \eqref{eq:po:od} is derivable from \eqref{eq:po:er} and the principle of \emph{composition} \eqref{eq:pocomposition} in the calculus $\mathsf{AX}$, so long as we can reason along the way about the potential response $Z_x$. By composition, for any $x$ and $y$ we have $y_x \wedge z_x^+ \rightarrow y_{z^+,x}$ and $y_x \wedge z_x^-\rightarrow y_{z^-,x}$, and by ER \eqref{eq:po:er} the latter gives $y_x \wedge z_x^-\rightarrow y_{z^+,x}$. As $Z_x$ is binary, we have $z_x^+ \vee z_x^-$, and thus by propositional reasoning, $y_x \rightarrow y_{z^+,x}$. The other direction $y_{z^+,x}\rightarrow y_x$ follows from the same argument by contraposition, as $Y$ too is binary.

Thus, while the full power of composition is not invoked, it is natural to read this example and much of the literature as implicitly assuming something like representability (thus implying composition).
Another source of support for this is that under representability one can show (see Prop.~\ref{prop:appx:itt1}) that $\text{ITT}_1 = \mathbf{E}(Y_{z^+} - Y_{z^-})$, an even simpler and manifestly identifiable expression for this average effect.

\end{example}

\begin{figure} \centering 
\begin{tabular}{ llllllll } 
& & & & $\RCM(\varepsilon)$ & & & \\
\toprule
$u$ & $P(u)$ & $X_{z^+}(u)$ & $X_{z^-}(u)$ & $Y_{x^+, z}(u)$ & $Y_{x^-, z}(u)$ & $Y_{x^+}(u)$ & $Y_{x^-}(u)$ \\
\midrule
$u_0$ & $\nicefrac{3}{4}-\varepsilon$ & $1$ & $0$ & $1$ & $0$ & $1$ & $0$ \\
$u_1$ & $\nicefrac{1}{4}-\varepsilon$ & $1$ & $0$ & $0$ & $1$ & $0$ & $1$ \\
$u_2$ & $ \varepsilon$ & $1$ & $0$ & $0$ & $0$ & $1$ & $0$ \\
$u_3$ & $ \varepsilon$ & $1$ & $0$ & $1$ & $1$ & $1$ & $1$ \\
\bottomrule
\end{tabular}
\caption{A family of RCMs $\{\RCM(\varepsilon)\}_{0 \le \varepsilon \le \nicefrac{1}{4}}$ such that $\RCM(\varepsilon)$ is representable iff $\varepsilon = 0$ (though \eqref{eq:pocomposition} and \eqref{eq:poreversibility} are met for any unit). Note that all members of this family meet \eqref{eq:pomonotonicity} and \eqref{eq:po:er} (the latter guaranteed since columns $Y_{x^+, z}(u)$, $Y_{x^-, z}(u)$ give the potential outcome for any $z$). Also, any experimentally testable quantities---including, in particular, $\text{ITT}_1$ and $\text{ITT}_2$---agree across the family, with $\frac{\text{ITT}_1}{\text{ITT}_2} = \nicefrac{1}{2}$. However the assumption in question \eqref{eq:po:od} holds in $\RCM(\varepsilon)$ iff $\varepsilon = 0$, and $\text{LATE} = \nicefrac{1}{2} + \varepsilon$, so that $\text{LATE} = \frac{\text{ITT}_1}{\text{ITT}_2}$ only when this holds, and it is not estimable in general.}
  \label{fig:example:cmreq}
\end{figure}

\subsection{Graphical Assumptions} \label{section:graph} As we saw above (Prop.~\ref{prop:scmrepresentation}), SCMs can be understood as \emph{representations} of suitable RCMs. As such, they also suggest further sources of  assumptions for deriving causal inferences. In particular, qualitative patterns of functional dependence introduce the possibility of \emph{graphical} methods:
\begin{definition}
    Let $\SCM = \big\langle \ExoVars, \EndoVars, \{f_V\}_{V \in \EndoVars}, P \big\rangle$ be an SCM.
    Then define the \emph{causal diagram} $\Graph(\SCM)$ of $\mathcal{M}$ as a graph over nodes $\EndoVars$,
    with mixed directed edges $\rightarrow$ and \emph{bidirected arcs} $\Confounded$.
    For any $V, V' \in \EndoVars$, there is a directed edge $V \rightarrow V'$ if $V \in \mathbf{Pa}_{V'}$,
    and there is a bidirected edge $V \Confounded V'$ if $\ExoVars_V, \ExoVars_{V'} \subset \ExoVars$ are correlated under $P$ (including if $\ExoVars_V \cap \ExoVars_{V'} \neq \varnothing$).
\end{definition}
Letting $\mathfrak{M}(\Graph)$ be the set of SCMs with diagram $\Graph$, we have $\mathfrak{M}(\Graph) \subset \mathfrak{M}_{\mathrm{uniq}}$ provided the directed edges in $\Graph$ form a dag (see Lem.~\ref{lem:app:recursive}). We thus assume this acyclicity of any $\Graph$.
When interpreting over an SCM, we include every possible potential outcome in $\Ll$.
Just as we earlier encoded assumptions about the potential outcomes of an RCM into $\mathcal{L}$, we may do the same for SCMs regarding their graphs.
A first observation is that Def.~\ref{def:equationaltranslation} translates axiom C6 of \cite{Halpern2000} to $\textsf{ProbRec}$ of \cite{ibelingicard2020}, thus rederiving the system $\textsf{AX}_3$ for the class of all acyclic SCMs, i.e.  $\bigcup_\Graph \mathfrak{M}(\Graph)$, from the latter.
We now encode the content of (the assumption of having) a \emph{particular} diagram $\Graph$ into $\mathcal{L}$.
Let $\mathbf{Pa}_V^\Graph = \{V'  \in \*V : V' \to V \in \Graph\}$ be the directed parents in a graph $\Graph$ of a vertex $V$.
We encode by way of two schemas, encapsulating what some \citep{pearlfundamentallaws} have called ``the two fundamental laws of causal inference'':
\begin{definition}\label{def:tgraph}
Let the {exclusion restriction} schema $\textsf{ER}^{\Graph}$ be the $\Lbase$ principle $y_{\mathbf{a}} \leftrightarrow  y_{\mathbf{p}}$,
for all variables $Y \in \*V$ and sets of variables $\mathbf{A} \supset \mathbf{Pa}_V^\Graph$,
where $y \in \Val(Y)$, $\mathbf{a} \in \Val(\mathbf{A})$, $\mathbf{p} = \pi_{\mathbf{Pa}_V^\Graph}(\mathbf{a})$.
Let the counterfactual independence schema $\textsf{cf-sep}^\Graph$ be, for all pairs of variable sets $\{Y_i\}_{1 \le i \le n}, \{Y'_j\}_{1 \le j\le n'} \subset \mathbf{V}$ such that there are no $Y_i, Y'_j$ for which $Y_i = Y'_j$ or $Y_i \Confounded Y'_j$ in $\Graph$,
\begin{eqnarray}\label{eqn:cfsep}
\textsf{cf-sep}^\Graph. && \mathbf{P}\big[\bigwedge_{1 \le i \le n} (y_i)_{\mathbf{p}_i} \land \bigwedge_{1 \le j \le n'} (y'_j)_{\mathbf{p}'_j}\big] = \mathbf{P}\big[\bigwedge_{1 \le i \le n} (y_i)_{\mathbf{p}_i}\big] \cdot \mathbf{P}\big[\bigwedge_{1 \le j \le n'} (y'_j)_{\mathbf{p}_j}\big]
\end{eqnarray}
where $y_i \in \Val(Y_i)$, $y'_j \in \Val(Y'_j)$, $\mathbf{p}_i \in \Val(\mathbf{Pa}^\Graph_{Y_i}), \mathbf{p}'_j \in \Val(\mathbf{Pa}^\Graph_{Y'_j})$ for each $Y_i$, $Y'_j$.
Then the translation of $\Graph$ is the combination of axioms $\mathrm{T}(\Graph) = \mathrm{T}(\textsf{ER}^\Graph) + \textsf{cf-sep}^\Graph$. \end{definition}
Note that while Ex.~\ref{eg:late} in no way relies on graphs, if we accept a $\Graph$ where $Z\not\rightarrow Y$, then $\textsf{ER}^\Graph$ yields $y_x \leftrightarrow y_{zx} \leftrightarrow y_{z'x}$ without further ado. Importantly, however, $\mathrm{T}\eqref{eq:pomonotonicity}$ is not valid over $\mathfrak{M}(\Graph)$ for any $\Graph$ containing the edge $Z \rightarrow X$, revealing an extra-graphical provenance.
On the other hand, $\textsf{cf-sep}$ is inexpressible in $\Lbase$---inferentially, the two approaches are incomparable.

A long-standing question has been whether exclusion restriction and independence axioms together could be \emph{complete},
in that they capture all the inferential content of a given causal diagram $\Graph$ (see, e.g., \cite{TianPearl2002,Evans2018,RichardsonEvans}). Answering such questions can help with the development of tractable inference methods. 
Partial completeness results for limited queries are known \citep{ShpitserPearlComplete}, and the method from Tian \cite{TianPearl2002} supplies an algorithm that is complete with respect to all equality constraints \cite{Evans2018}. 
Placing no limitations on queries beyond their expressibility in $\Ll$---and thus including inequality constraints as well---but making certain restrictions on $\Graph$, we answer this question in the affirmative:
\begin{theorem}\label{thm:graphicalsystem}
    For any acyclic diagram $\Graph$, axioms $\mathrm{T}(\Graph) + \textsf{SCM}$ are sound for $\mathcal{L}$ over $\mathfrak{M}(\Graph)$,
    and also complete if the bidirected arcs in $\Graph$ form a disjoint union of complete graphs.
\end{theorem}
Often the famous \emph{d-separation} conditional independence criterion (Def.~\ref{def:app:dsep}) is used in place of our $\textsf{cf-sep}$. Since all instances of the latter are instances of the former, our Thm.~\ref{thm:graphicalsystem} is stronger (see Cor.~\ref{cor:app:dsepcomplete}).
This completeness result implies that for such a $\Graph$, any known graphical conclusions---including $do$-calculus, identifiability results, and bounds---can be rederived in our calculus, e.g.:
\begin{example}[Verma constraints]\label{ex:verma}

We derive the \emph{Verma constraint} \cite{VermaPearl1990,TianPearl2002} over the graph $\Graph$ of Fig.~\ref{fig:gverma}a that $\sum_w \mathbf{P}(y \mid z, w, x) \mathbf{P}(w \mid x)$ does not depend functionally on $x$:
\begin{multline*}
\sum_w \frac{\mathbf{P}(y, z, w, x) \mathbf{P}(w, x)}{\mathbf{P}(z, w, x) \mathbf{P}(x)}
\overset{\mathrm{T}(\textsf{C})}{=}
\sum_w \frac{\mathbf{P}(y_{zwx}, z_{ywx}, w_{yzx}, x_{yzw}) \mathbf{P}(w_x, x_w)}{\mathbf{P}(z_{wx}, w_{zx}, x_{zw}) \mathbf{P}(x)}\\
\overset{\mathrm{T}(\textsf{ER}^\Graph)}{=}
\sum_{w} \frac{\mathbf{P}(y_z, z_w, w_{x}, x) \mathbf{P}(w_x, x)}{\mathbf{P}(z_w, w_{x}, x) \mathbf{P}(x)}
\overset{\eqref{eqn:cfsep}}{=} \sum_{w} \frac{\mathbf{P}(y_z, w_x) \cancel{\mathbf{P}(z_w) \mathbf{P}(x) \mathbf{P}(w_x) \mathbf{P}(x)}}{\cancel{\mathbf{P}(z_w)\mathbf{P}(w_x)\mathbf{P}(x) \mathbf{P}(x)}}
=
\mathbf{P}(y_z).
\end{multline*}

\end{example}
To summarize, for a wide class of cases, Thm.~\ref{thm:graphicalsystem} tells us that the two principles encoded in $\mathrm{T}(\mathcal{G})$ exhaust the types of algebraic statements that a researcher is committed to when venturing a graphical assumption. Putting this in algebraic terms facilitates a perspective on such assumptions that can be naturally interpreted with respect to RCMs, independent of any representation by an SCM.

\subsection{Single-World Intervention Graphs}
Another graphical framework that draws  on ideas and concepts from both frameworks is that of \emph{single-world intervention graphs} (SWIGs) \cite{RichardsonRobins}. In comparison to the usual formulation of do-calculus, SWIGs facilitate reasoning with a wider class of expressions by combining graphical and potential outcome notations (see, e.g., \cite{pmlr-v89-malinsky19b}). 
Here we show that this ``hybrid'' framework can also naturally be assimilated to the logical perspective adopted in the present paper. Assuming acyclicity:

\begin{definition}
Let $\mathcal{D}$ be a dag over $\mathbf{V}$ and let $\mathbf{x}$ be an intervention.
Let $\mathbf{An}_V^{\mathcal{D}} = \{ V' \in \mathbf{V} : V' \to \dots \to V \text{ in } \mathcal{D} \}$ be the directed ancestors\footnote{I.e., the transitive closure of the converse of the parent sets $\mathbf{Pa}_V^{\mathcal{D}}$. Note that $V \in \mathbf{An}_V^{\mathcal{D}}$ and $\mathbf{Pa}_V^{\mathcal{D}} \subset \mathbf{An}_V^{\mathcal{D}}$.} of a variable $V$ in $\mathcal{D}$.
Then the SWIG $\mathcal{D}_{\mathbf{x}}$ has nodes labeled $\{ V_{\pi_{\mathbf{A}}(\mathbf{x})} : V \in \mathbf{V} \} \cup 
\{ \pi_{V}(\mathbf{x}) : V \in \mathbf{X}\}$, where $\mathbf{A} = (\mathbf{An}^{\mathcal{D}}_V \cap \mathbf{X})\setminus \{V\}$, with
nodes in the first set in this union termed \emph{random}, and those in the second termed \emph{fixed}.
This SWIG has edges
\begin{align*}
\{\pi_X(\mathbf{x}) \to V_{\mathbf{p}} : X \in \mathbf{X}, \Graph \text{ has edge } X \to V \} \cup \{V_{\mathbf{p}} \to V'_{\mathbf{p}'} : V \notin \mathbf{X}, \Graph \text{ has edge }V \to V' \}.
\end{align*}
\end{definition}
\begin{figure} \centering 
\subfigure[Mixed Diagram $\Graph$] {
\begin{tikzpicture}[SCM,scale=1, every node/.append style={transform shape}]
    \node (X) at (0,1) {$X$};
    \node (W) at (0,0) {$W$};
    \node (Z) at (1.25,0) {$Z$};
    \node (Y) at (2.5,0) {$Y$};

    \path [->] (X) edge (W);
    \path [->] (W) edge (Z);
    \path [->] (Z) edge (Y);
    \path [<->] (W) edge [dashed, bend left=60] (Y);
\end{tikzpicture}
}\hspace{0.25in}
\subfigure [Directed Graph $\mathcal{D}$] {
\begin{tikzpicture}[SCM,scale=1, every node/.append style={transform shape}]
    \node (X) at (0,1) {$X$};
    \node (W) at (0,0) {$W$};
    \node (Z) at (1.25,0) {$Z$};
    \node (Y) at (2.5,0) {$Y$};
    \node (U) at (1.25,1) {$U$};

    \path [->] (X) edge (W);
    \path [->] (W) edge (Z);
    \path [->] (Z) edge (Y);
    \path [->] (U) edge (W);
    \path [->] (U) edge (Y);
\end{tikzpicture}
   }\hspace{0.25in}
\subfigure [SWIG $\mathcal{D}_{W = w^+}$] {
\begin{tikzpicture}[SCM,scale=1, every node/.append style={transform shape}]
    \node (X) at (-1.25,1) {$X_\varnothing$};
    \node (W) at (-1.25,0) {$W_\varnothing$};
    \node (WF) at (0,0) {$w^+$};
    \node (Z) at (1.25,0) {$Z_{w^+}$};
    \node (Y) at (2.5,0) {$Y_{w^+}$};
    \node (U) at (1.25,1) {$U_\varnothing$};

    \path [->] (X) edge (W);
    \path [->] (WF) edge (Z);
    \path [->] (Z) edge (Y);
    \path [->] (U) edge (W);
    \path [->] (U) edge (Y);
\end{tikzpicture}
}
    \caption{(a) Graph $\Graph$ for Ex.~\ref{ex:verma}. Note that we have $\textsf{ER}^\Graph \vDash y_{zwx} \leftrightarrow y_z$, $z_{ywx} \leftrightarrow z_w$, $w_{yzx} \leftrightarrow w_x$, $x_{yzw} \leftrightarrow x$, $x_w \leftrightarrow x$, $z_{wx} \leftrightarrow z_w$, etc., while pairs where $\textsf{cf-sep}^\Graph$ is applicable include $\{W, Y\}$, $\{X, Z\}$ and $\{Z\}, \{W, X\}$ (and any pairwise subsets thereof). These instances are used in Ex.~\ref{ex:verma}. \label{fig:gverma}\\
    (b) Directed graph $\mathcal{D}$, obtained from $\Graph$ by making the latent exogenous variable $U$ explicit.\\
    (c) SWIG obtained from $\mathcal{D}$ for the intervention $W = w^+$, with fixed node $w^+$. \eqref{eq:swsep} implies $\mathbf{P}(w, y_{w^+}, z_{w^+}) \mathbf{P}(u) = \mathbf{P}(w, u) \mathbf{P}(y_{w^+}, z_{w^+}, u)$ since  every path between $W_\varnothing$ and either $Z_{w^+}$ or $Y_{w^+}$ goes through $U_\varnothing$, where the directions of the arrows do not collide.}\label{fig:swig}
\end{figure}
See Fig.~\ref{fig:swig}c for an example of this construction. Note that edges in the first set in the union above have fixed heads and random tails, while those in the second set have random heads and tails.

\begin{definition}
Define the following conditional independence schema $\textsf{sw-sep}^{\mathcal{D}_{\mathbf{x}}}$,
for any sets of random nodes $\{ (X_i)_{\mathbf{p}_i} : 1 \le i \le l \}$ and $\{ (Y_j)_{\mathbf{p}'_j} : 1 \le j \le m \}$ that are d-separated (Def.~\ref{def:app:dsepdir}) given $\{ (Z_k)_{\mathbf{p}''_k} : 1 \le k \le n \}$ in the SWIG $\mathcal{D}_{\mathbf{x}}$: 
\begin{multline}
\textsf{sw-sep}^{\mathcal{D}_{\mathbf{x}}}. \quad \mathbf{P}\big[\bigwedge_{\substack{1 \le i \le l\\1 \le j \le m}} (x_i)_{\mathbf{p}_i} \land (y_j)_{\mathbf{p}'_j}\big] \cdot
\mathbf{P}[\bigwedge_{1 \le k \le n} (z_k)_{\mathbf{p}''_k}]\\
= \mathbf{P}\big[\bigwedge_{\substack{1 \le i \le l\\1 \le k \le n}} (x_i)_{\mathbf{p}_i} \land (z_k)_{\mathbf{p}''_k}\big] \cdot \mathbf{P}\big[\bigwedge_{\substack{1 \le j \le m\\1 \le k \le n}} (y_j)_{\mathbf{p}'_j} \land (z_k)_{\mathbf{p}''_k}\big].
\label{eq:swsep}\end{multline}
\end{definition}
One notable model associated with SWIGs is the \emph{FFRCISTG} \cite{ROBINS19861393}; given the same graph, FFRCISTGs are compatible with SCMs \cite{RichardsonRobins} but issue fewer (potentially controversial) implications:
\begin{definition}
Let $\mathcal{R}$ be a full RCM. Then $\mathcal{R}$ is a \emph{FFRCISTG over $\mathcal{D}$} if every instance of $\mathrm{T}(\textsf{ER}^{\mathcal{D}})$ and $\textsf{sw-sep}^{\mathcal{D}}$ holds in its counterfactual distribution. Let $\mathfrak{F}(\mathcal{D})$ be the class of FFRCISTGs over $\mathcal{D}$.\label{def:ffrcistg}
\end{definition}
\begin{proposition}
Suppose the SCM $\mathcal{M} \in \mathfrak{M}(\mathcal{D})$ represents the full RCM $\mathcal{R}$. Then $\mathcal{R} \in \mathfrak{F}(\mathcal{D})$.
\qed
\end{proposition}
Given that Def.~\ref{def:ffrcistg} already defines $\mathfrak{F}(\mathcal{D})$ in terms of $\mathcal{L}$-principles, while \cite{RichardsonRobins} have shown the soundness direction, the following is straightforward:
\begin{theorem}
$\textsf{RCM} + \mathrm{T}(\textsf{ER}^{\mathcal{D}}) + \bigcup_{\mathbf{x}}\textsf{sw-sep}^{\mathcal{D}_{\mathbf{x}}}$ is sound and complete over $\mathfrak{F}(\mathcal{D})$. \label{thm:swig}
\qed
\end{theorem}

\section{Conclusion}
The task of this paper has been to clarify the senses in which the Rubin causal model and structural causal models are very closely related formalisms for encoding causal assumptions and deriving causal conclusions. We concur with \cite{MorganWinship}, \cite{Imbens2020}, \cite{Weinberger2022} and others that ``there are insights that arise when using each that are less transparent when using the other'' \citep[p. 8]{Weinberger2022}. Our interest in this paper has been to elucidate the comparison from a theoretical (``in principle'') perspective. 

We do not suppose that the present work will be the final word on theoretical connections between RCMs and SCMs. On the contrary, there remain numerous open questions. Perhaps chief among these is the generalization of Thm.~\ref{thm:graphicalsystem} to encompass all possible causal diagrams (not just those in which the bidirected arcs form a disjoint union of complete graphs). Does the theorem hold with no further principles, or do additional algebraic constraints arise? This important open question \cite{TianPearl2002,Evans2018,RichardsonEvans} is a crucial step toward a complete theoretical synthesis of the two frameworks. 

\subsubsection*{Acknowledgments} This work was partially supported by a seed grant from the Stanford  Institute for Human-Centered Artificial Intelligence. We also thank Elias Bareinboim and Guido Imbens for helpful conversations. 


{
\small

\bibliographystyle{abbrvnat}
\bibliography{uai2023}
}

\newpage
\section*{Appendices}
\appendix
These appendices contain demonstrations of the results in the main text as well as additional technical notes.

\section{Modeling (\S{}1)}\label{s:app:modeling}

\subsection{Preliminaries}

\begin{proposition}\label{prop:app:limitofunits}
Uniform RCMs are dense (in the weak subspace topology on counterfactual distributions\footnote{See \cite{NEURIPS2021_2c463dfd, Billingsley} for more detail on this topology and inducing metrics.}): for any $\RCM$, $\varepsilon > 0$, there is a $\RCM'$ whose distribution is uniform that $\varepsilon$-approximates $\RCM$ (in, e.g., the standard metric on counterfactual distributions considered as points in Euclidean space, or the L\'{e}vy-Prohorov metric).
\end{proposition}
\begin{proof}
Follows by density of $\mathbb{Q} \subset \mathbb{R}$ since any RCM with a rational distribution is representable also by a uniform distribution on units.
\end{proof}
\begin{proof}[Proof of Prop.~\ref{prop:scmrepresentation}]
Suppose $\mathcal{R} = \langle \mathcal{U}, \*V, \mathfrak{O}, \{f_{Y_\Ante}\}_{Y_\Ante \in \mathfrak{O}}, P\rangle$ is representable. Then we have an SCM $\SCM = \langle \*U, \*V, \{f_V\}_V, P'\rangle \in \mathfrak{M}_{\mathrm{uniq}}$ inducing its counterfactual distribution on $\Val(\mathfrak{O})$.
For each $u \in \mathcal{U}$, there must be a ${\*u}_{u} \in \Val(\*U)$ such that $P'({\*u}_{u}) > 0$ and $\mathcal{M}_\Ante, {\*u}_{u} \models Y = f_{Y_\Ante}(u)$ for each $Y_\Ante \in \mathfrak{O}$.
Consider the RCM $\RCM' = \langle \mathcal{U}, \*V, \{Y_\Ante\}_{\text{all } Y_\Ante}, \{f'_{Y_\Ante}\}_{Y_\Ante}, P\rangle$ where we define $f'_{Y_\Ante}(u) = f_{Y_\Ante}(u)$ for each $Y_\Ante \in \mathfrak{O}$, and $f'_{Y_\Ante} = \pi_Y(\*v)$ for each $Y_\Ante \notin \mathfrak{O}$ where $\*v$ is the unique solution such that $\SCM_\Ante, \mathbf{u}_u \models \*v$.
Note that $\RCM'$ extends $\RCM$ by construction, has no proper extension itself, and satisfies composition and reversibility by the soundness direction of \cite[Thm.~3.2]{Halpern2000}, as these principles must apply to $\SCM$.

Conversely suppose we have the extension $\RCM'$. Then by the completeness direction of the proof of \cite[Thm.~3.2]{Halpern2000} we can construct an SCM $\SCM \in \mathfrak{M}_{\mathrm{uniq}}$ representing $\RCM'$ and thus also $\RCM$. Specifically, we can derive a unique maximal consistent set $\Gamma$ from the outcomes defined in $\RCM'$ (unique maximality is guaranteed since $\RCM'$ has no proper extension and consistency since $\RCM'$ meets composition and reversibility), and from there define the equations of $\SCM$.
\end{proof}

\subsection{Causal Abstraction}

In this section we will use the following technical result to decompose constructive partial translations.
\begin{lemma}\label{lem:parttransdecomp}
Suppose $\tau$ is a constructive translation of $\EndoVarsLowLevel$ to $\EndoVarsHighLevel$ and
let ${\*x}_\LowLevel, {\*x}_\HighLevel$ be partial settings of ${\*X}_\LowLevel \subset {\*V}_\LowLevel, {\*X}_\HighLevel \subset {\*V}_\HighLevel$ respectively.
Then $\tau({\*x}_\LowLevel) = {\*x}_\HighLevel$ iff:
\begin{enumerate}
\item For each $X_\HighLevel \in {\*X}_\HighLevel$, we have that $\tau_{X_\HighLevel}(\theta_{X_\HighLevel}) = \pi_{X_\HighLevel}({\*x}_\HighLevel)$ where we define $\theta_{V_\HighLevel} = \{\*x \in \Val(\Pi_{V_\HighLevel}) : \pi_{{\*X}_\LowLevel \cap \Pi_{V_\HighLevel}}(\*x) = \pi_{{\*X}_\LowLevel \cap \Pi_{V_\HighLevel}}({\*x}_\LowLevel)\}$ for each $V_\HighLevel \in {\*V}_\HighLevel$.\label{lem:parttransdecomp:1}
\item For each $V_\HighLevel \notin {\*X}_\HighLevel$, we have that $\tau_{V_\HighLevel}(\theta_{V_\HighLevel}) = \Val(V_\HighLevel)$.\label{lem:parttransdecomp:2}
\end{enumerate}
\begin{proof}
First we show the ``if'' direction. It follows easily from \ref{lem:parttransdecomp:1} that $\tau({\*v}_\LowLevel) \in \pi^{-1}_{{\*X}_\HighLevel}({\*x}_\HighLevel)$ for any ${\*v}_\LowLevel \in \pi^{-1}_{{\*X}_\LowLevel}({\*x}_\LowLevel)$, and it follows from \ref{lem:parttransdecomp:2} and \ref{lem:parttransdecomp:1} that for any ${\*v}_\HighLevel \in \pi^{-1}_{{\*X}_\HighLevel}({\*x}_\HighLevel)$ there is a ${\*v}_\LowLevel \in \pi^{-1}_{{\*X}_\LowLevel}({\*x}_\LowLevel)$ with $\tau({\*v}_\LowLevel) = {\*v}_\HighLevel$ (specifically, any extension of $\Ante_\LowLevel$ works).

Now we show the ``only if'' direction. To show \ref{lem:parttransdecomp:1}: suppose not. Then since $\theta_{V_\HighLevel}$ is always nonempty there is some ${\*x} \in\Val(\Pi_{X_\HighLevel})$ such that $\pi_{{\*X}_\LowLevel \cap \Pi_{X_\HighLevel}}(\*x) = \pi_{{\*X}_\LowLevel \cap \Pi_{X_\HighLevel}}({\*x}_\LowLevel)$ and $\tau_{X_\HighLevel}(\*x) \neq \pi_{X_\HighLevel}({\*x}_\HighLevel)$.
But then we cannot have $\tau({\*v}_\LowLevel) \in \pi^{-1}_{{\*X}_\HighLevel}({\*x}_\HighLevel)$ for any ${\*v}_\LowLevel \in \pi^{-1}_{{\*X}_\LowLevel}({\*x}_\LowLevel) \cap \pi^{-1}_{\Pi_{X_\HighLevel}}({\*x})$, which is nonempty.
To show \ref{lem:parttransdecomp:2}: suppose not. Then there is some $v_\HighLevel \in \Val(V_\HighLevel)$ such that there is no ${\*x} \in \Val(\Pi_{V_\HighLevel})$ with $\pi_{{\*X}_\LowLevel\cap \Pi_{V_\HighLevel}}(\Ante) = \pi_{{\*X}_\LowLevel \cap \Pi_{V_\HighLevel}}(\Ante_\LowLevel)$ and $\tau_{V_\HighLevel}(\Ante) = v_\HighLevel$.
But then for any ${\*v}_\HighLevel \in \pi^{-1}_{{\*X}_\HighLevel}(\Ante_\HighLevel) \cap \pi^{-1}_{V_\HighLevel}(v_\HighLevel)$, which is nonempty since $V_\HighLevel \notin {\*X}_\HighLevel$, there cannot be any ${\*v}_\LowLevel \in \tau^{-1}({\*v}_\HighLevel) \cap \pi^{-1}_{{\*X}_\LowLevel}(\Ante_\LowLevel)$, contradicting that $\tau(\Ante_\LowLevel) = \Ante_\HighLevel$.
\end{proof}
\end{lemma}
\begin{proof}[Proof of Prop.~\ref{prop:abstractionproperties}]
Suppose $\mathcal{H}$ is not effective so there is some $\{ (\mathbf{y}^j_\HighLevel)_{\Ante^j_\HighLevel}\}_{1\le j \le n}$ in its counterfactual support such that for some $k$ and $Y$ we have that $\pi_Y(\mathbf{y}^k_\HighLevel) \neq \pi_Y(\Ante^k_\HighLevel)$. This means there is some $ (\mathbf{y}^l_\LowLevel)_{\Ante^l_\LowLevel}$ in a counterfactual in the support of $\mathcal{L}$ with $\tau(\mathbf{y}^l_\LowLevel) = \mathbf{y}^k_\HighLevel$ and $\tau(\Ante^l_\LowLevel) = \Ante^k_\HighLevel$. By Lem.~\ref{lem:parttransdecomp}, \ref{lem:parttransdecomp:1}, 
there is no $\*x \in \Val(\Pi_{Y})$ such that $\pi_{{\*Y}^l_\LowLevel \cap \Pi_{Y}}(\*x) = \pi_{{\*Y}^l_\LowLevel \cap \Pi_{Y}}(\mathbf{y}^l_\LowLevel)$ and $\pi_{{\*X}^l_\LowLevel \cap \Pi_{Y}}(\*x) = \pi_{{\*X}^l_\LowLevel \cap \Pi_{Y}}(\mathbf{x}^l_\LowLevel)$. This implies that there is some $Y' \in {\*X}^l_\LowLevel \cap {\*Y}^l_\LowLevel \cap \Pi_{Y}$ such that $\pi_{Y'}({\*x}^l_\LowLevel) \neq \pi_{Y'}({\*y}^l_\LowLevel)$, contradicting the effectiveness of $\mathcal{L}$.

As for the second claim,
by induction it suffices to show this for the case where $\mathcal{H}'$ is identical to $\mathcal{H}$ except that it lacks a single potential outcome $Y_{\Ante}$.
Consider the model $\mathcal{L}'$ formed applying the mapping $\big\{ (\mathbf{y}^i_\LowLevel)_{\Ante^i_\LowLevel} \big\}_{1 \le i \le m} \mapsto \big\{ (\mathbf{z}^i_\LowLevel)_{\Ante^i_\LowLevel} \big\}_{1\le i \le m}$ of counterfactuals to $\mathcal{L}$, where $\mathbf{z}^i_\LowLevel = \pi_{\mathbf{Y}^i_\LowLevel \setminus \Pi_Y}(\mathbf{y}^i_\LowLevel)$ if $\tau(\Ante^i_\LowLevel) = \Ante$ and $\mathbf{z}^i_\LowLevel = \mathbf{y}^i_\LowLevel$ otherwise.
It follows from Lem.~\ref{lem:parttransdecomp} that $\mathcal{H}' \prec_\tau \mathcal{L}'$.
\end{proof}

\begin{proof}[Proof of Thm.~1]
Let $\RCM$ be as in Def.~\ref{def:rcm} and index its potential outcomes $\mathfrak{O}$ as
$\big\{ (\mathbf{Y}^i)_{\Ante^i} \big\}_{1 \le i \le n} \cup \{ (\mathbf{Y}^{n+1})_{\Ante^{n+1}} \}$
where $\Ante^i \neq \varnothing$ for any $1 \le j \le n$ and $\Ante^{n+1} = \varnothing$; here $\varnothing$ represents the empty intervention on $\mathbf{V}$.
We assume that $n > 0$, since otherwise $\RCM$ is trivially representable.
The potential response functions of $\RCM$ are $\{f_{Y_\Ante}\}_{Y_\Ante}$.

Define a low-level set of variables $\EndoVars_\LowLevel = \{ (V, j) : \substack{V \in \*V \\ 1 \le j \le n + 1} \}$ with $\Val(V, j) = \Val(V) \cup \{ \star\}$ where $\star \notin \Val(V)$ for each $(V, j) \in \mathbf{V}_\LowLevel$.
Define $\tau$ as a constructive translation with a partition $\Pi$, defined by $\Pi_{V} = \big\{ (V, j) \big\}_{1 \le j \le n + 1}$ for each $V \in \*V$, and $\tau_V(\mathbf{p}_\LowLevel) = p_j$ iff $\mathbf{p}_\LowLevel \in \Val(\Pi_V)$ is such that there is exactly one $j$ such that $\pi_{(V, j)}(\mathbf{p}_\LowLevel) \neq \star$, with $p_j$ this value. Let 
\begin{multline}
\mathfrak{O}_\LowLevel = \big\{ (Y, i)_{\Ante^i_\LowLevel} : \substack{1 \le i \le n\\ Y \in \mathbf{Y}^i} \big\} \cup \{ (Y, n+1)_{\varnothing_\LowLevel} : Y \in \mathbf{Y}^{n+1} \} \cup \Big\{ (Y, i)_{\Ante^j_\LowLevel}: \substack{1 \le i, j \le n\\i \neq j\\Y \in \mathbf{Y}^i} \Big\}\\ \cup \big\{ 
(Y, n+1)_{\Ante^i_\LowLevel} : \substack{1 \le i \le n\\ Y \in \mathbf{Y}^{n+1}} \big\} \cup \big\{ (Y, i)_{\varnothing_\LowLevel} : \substack{1 \le i \le n\\ Y \in \mathbf{Y}^i} \big\}
\label{eq:polowerlevel}
\end{multline}
where for each $1 \le i \le n$, we let $\mathbf{X}^i_\LowLevel = \{(X, i) : X \in \mathbf{X}^i\}$ with $\pi_{(X, i)}(\Ante^i_{\LowLevel}) = \pi_X(\Ante^i)$ for every $X \in \mathbf{X}^i$, $\varnothing_\LowLevel$ is an empty intervention on $\*V_{\LowLevel}$, and define a set of potential responses $\mathfrak{F}_\LowLevel = \{f^\LowLevel_{Z_\Ante}\}_{Z_\Ante}$ for each of these outcomes via
\begin{equation}
f^\LowLevel_{Z_\Ante} = \begin{cases}
f_{Y_{\Ante^i}}, Z_\Ante = (Y, i)_{\Ante^i_\LowLevel}\\
f_{Y_{\Ante^{n+1}}}, Z_\Ante = (Y, n+1)_{\varnothing_\LowLevel}\\
f_\star, \text{otherwise}
\end{cases}
\label{eq:prlowerlevel}
\end{equation}
where $f_\star$ is a constant function with $f_\star(u) = \star$ for any $u \in \mathcal{U}$.

We claim that $\RCM_\LowLevel = \langle \mathcal{U}, \EndoVars_\LowLevel, \mathfrak{O}_\LowLevel, \{f^\LowLevel_{Z_\Ante}\}_{Z_\Ante}, P \rangle$, where $\RubinUnits$ and $P$ are the same as those in $\RCM$, abstracts to $\RCM$ and is representable.
To show the former, we employ the following result:
\begin{lemma}
Let $\Ante_\LowLevel \in \Val(\mathbf{X}_\LowLevel)$ for some $\mathbf{X}_\LowLevel \subset \mathbf{V}_\LowLevel$
and suppose that there is some $\mathbf{X} \subset \mathbf{V}$ and $j \in \{1, \dots, n+1\}$ such that $\pi_{(V, i)}(\Ante_\LowLevel) \neq \star$ iff $V \in \mathbf{X}$ and $i = j$; suppose further that for each $V$, there is at least one $i \in \{1, \dots, n+1\}$ such that $(V, i) \notin \mathbf{X}_\LowLevel$ or $\pi_{(V, i)}(\Ante_\LowLevel) \neq \star$.
Then $\tau(\Ante_\LowLevel) = \mathbf{x}_\HighLevel$ where $\mathbf{X}_\HighLevel = \mathbf{X}$ and $\pi_X(\mathbf{x}_\HighLevel) = \pi_{(X, i)}(\Ante_\LowLevel)$ for each $X \in \mathbf{X}$.
\label{lem:partialtranslationslowerlevel}
\end{lemma}
\begin{proof}
This follows directly from the construction of $\tau$, in light of Lem.~\ref{lem:parttransdecomp}.
\end{proof}
Lem.~\ref{lem:partialtranslationslowerlevel} implies that $\tau(\varnothing_\LowLevel) = \varnothing$, $\tau(\Ante^i_\LowLevel) = \Ante^i$ for each $i$.
Further, for each $u \in \mathcal{U}$ and $i$ it implies that $\tau(\mathbf{y}^\LowLevel_i(u)) = \mathbf{y}_i(u)$
where $\mathbf{y}^\LowLevel_i(u)$ is the partial setting of $\mathbf{V}_\LowLevel$ induced by $\RCM_\LowLevel$ for potential outcomes under $\Ante_\LowLevel^i$ for unit $u$, viz., by \eqref{eq:polowerlevel} and \eqref{eq:prlowerlevel},
that with the projections
$\pi_{(Y, i)}(\mathbf{y}^\LowLevel_i(u)) = f_{Y_{\Ante^i}}(u)$ for each $Y \in \mathbf{Y}^i$,
$\pi_{(Y, j)}(\mathbf{y}^\LowLevel_i(u)) = \star$ for each $j \neq i$, $Y \in \mathbf{Y}^j$;
and $\mathbf{y}_i(u)$ is the analogue in $\RCM$, with $\pi_{Y}(\mathbf{y}_i(u)) = f_{Y_{\Ante^i}}(u)$ for $Y \in \mathbf{Y}^i$.
These facts establish that $P^{\RCM}_\counterfactual = \tau_*(P^{\RCM_\LowLevel}_\counterfactual)$ so that $\RCM \prec_\tau \RCM_\LowLevel$.

Next we demonstrate that $\RCM_\LowLevel$ is representable by constructing an explicit SCM representation $\SCM_\LowLevel = \langle \{U\}, \*V_\LowLevel, \{e_{V_\LowLevel}\}_{V_\LowLevel \in \EndoVars_\LowLevel}, P \rangle$, where $U$ is an exogenous variable with $\Val(U) = \RubinUnits$, while $P$ is the same distribution as in $\RCM$ and $\RCM_\LowLevel$.
It will be convenient to set up $\SCM_\LowLevel$ so as to be \emph{recursive}.
We say that an $\SCM$ over endogenous variables $\EndoVars$ is recursive if there is a total order $<$ on $\*V$ such that $\mathbf{Pa}_V \subset \{ V' \in \*V : V' < V\}$ for every $V \in \*V$.
Recursiveness guarantees uniqueness and existence of solutions under every unit and intervention, regardless of the particular structural functions $\{f_V\}_V$ and merely by virtue of their signatures:
\begin{lemma}\label{lem:app:recursive}
  If $\SCM$ is recursive then there is a unique $\*v \in \Val(\*V)$ such that $\mathcal{M}_{\Ante}, \Unit \models \*v$ for any $\*x \in \bigcup_{\mathbf{X} \subset \EndoVars}\Val(\*X)$, unit $\Unit$.
\end{lemma}
\begin{proof}
  Assume without loss $\*V = \{ V_1, \dots, V_m \}$ with $V_1 < \dots < V_m$ and prove by induction, the inductive hypothesis being that there is a unique $\*v$ over $\{V_1, \dots, V_j\}$ for $j \le m$, which extends as $\pi_{V_{j+1}}(\*v) = f_{V_{j+1}}\big(\*u, \pi_{\{V_1, \dots, V_{j}\}}(\*v)\big)$ if $V_{j+1} \notin \mathbf{X}$ and $\pi_{V_{j+1}}(\Ante)$ otherwise.
\end{proof}
For the recursive order $<$, pick any $<$ such that $(V, h) < (V', h')$ for any $V, V' \in \EndoVars$ and $1 \le h, h' \le n+1$ whenever $h < h'$,
for any $1 \le j \le n$, we have that $(X, j) < (Y, j)$ for any $X \in \EndoVarsX^j \setminus \mathbf{Y}^j$, $Y \in \mathbf{Y}^j$, and $(Y_l, n+1) < (Y_{l'}, n+1)$
whenever $l < l'$, where we fix an arbitrary indexing $\{Y_l: 1 \le l \le |\mathbf{Y}^{n+1}|\}$ of $\mathbf{Y}^{n+1}$.
Let $\mathbf{Pa}_{V_\LowLevel} = \{V' \in V_\LowLevel: V' < V_\LowLevel\}$, $\mathbf{U}_{V_\LowLevel} = \mathbf{U}$ for each $V_\LowLevel$
and define the structural function $e_{(V, h)}$ by
\begin{equation}\label{eq:sflowlevel}
e_{(V, h)}(u, \*p) = \begin{cases}
    f_{V_\varnothing}(u), & h = n+1, V \in \mathbf{Y}^{n+1}, \*p \text{ is such that } \pi_{(Y_{l'}, n+1)}(\*p) = f_{V_\varnothing}(u) \\
    &\qquad \text{ for any } l' \text{ such that } n+1 \le l' < l \text{ and }\\
    &\qquad \pi_{(V', h')}(\*p) = \star \text{ for any other } (V', h')\\
f_{V_{\Ante^h}}(u), & 1\le h \le n, V \in \mathbf{Y}^h, \big|\EndoVarsX^h \setminus\mathbf{Y}^h\big| > 0, \*p \text{ is such that }\\
    &\qquad
    \pi_{(X, h)}(\*p) =  \pi_X(\Ante_h) \text{ for any } \\
    &\qquad X \in \EndoVarsX^h \setminus\mathbf{Y}^h
    \text{ and }\pi_{(V', h')}(\*p) = \star \text{ for any other } (V', h') \\
\star, & \text{otherwise}
\end{cases}.
\end{equation}
We can find the unique solution of $\SCM_\LowLevel$ under any of our interventions of interest by following these equations in order:
\begin{lemma}[Solutions of $\SCM_\LowLevel$]
Let $\Ante \in \{\Ante^j_{\LowLevel} \}_{1 \le j \le n} \cup \{\varnothing_\LowLevel\}$ and $i \in \RubinUnits$.
Then $(\SCM_\LowLevel)_{\Ante}, i \models \*v \in \Val(\EndoVars_\LowLevel)$ where
   \begin{align}\label{eq:lowlevelsolution}
\pi_{(V, h)}(\*v) = \begin{cases}
             \pi_V(\Ante_j) & \Ante = \Ante^j_{\LowLevel} \text{ for some } j, V \in \EndoVarsX^j \setminus \mathbf{Y}^j, h = j \\
              f_{V_{\Ante^j}}(i) & \Ante = \Ante^j_{\LowLevel} \text{ for some } j, V \in \mathbf{Y}^j, h = j \\
             f_{V_\varnothing}(i) & \Ante = \varnothing_\LowLevel, h = n+1, V \in \mathbf{Y}^{n+1}\\
             \star & \text{otherwise}
        \end{cases}.
   \end{align}
\end{lemma}
\begin{proof}
    Suppose $\Ante = \Ante^j_{\LowLevel}$ for some $j \le n$. We claim that $\pi_{(V, k)}(\*v) = \star$ for any $V$ and $k < j$.
    The only way it could be otherwise is in the second case of \eqref{eq:sflowlevel}, but since $k \neq j$ and $\big|\mathbf{X}^k \setminus \mathbf{Y}^k\big| > 0$ there is here at least one $(X, k) < (Y^k, k)$ with $\pi_{(X, k)}(\*v) = \star$, so that this case in fact cannot apply.
    Now for $k = j$, it is clear by construction of $\Ante^j_{\LowLevel}$ that $\pi_{(V, j)}(\*v) = \pi_V(\Ante^j)$ for each $V \in \EndoVarsX^j \setminus \mathbf{Y}^j$. For $V \in \mathbf{Y}^j$, if $V \in \EndoVarsX^j$ then by effectiveness and the construction of $\Ante^j_{\LowLevel}$, we have that $\pi_{(V, j)}(\*v) = \pi_{V}(\Ante^j) = f_{V_{\Ante^j}}(i)$. Otherwise, since $\EndoVarsX^j \neq \varnothing$ and given what we have already found about $\*v$ we fall into the second case in \eqref{eq:sflowlevel}. Thus $\pi_{(V, j)}(\*v) = f_{V_{\Ante^j}}(i)$.
    For $k > j$, if $k \le n$ then the second case of \eqref{eq:sflowlevel} cannot apply so $\*v$ extends only by $\star$'s, and if $k = n+1$ then since $\Ante^j \neq \varnothing$ the first case of \eqref{eq:sflowlevel} likewise cannot apply, and we get only $\star$'s. This gives us precisely the $\*v$ specified by the first, second, and fourth cases in \eqref{eq:lowlevelsolution}.

    Next suppose $\Ante = \varnothing_\LowLevel$. Then $\pi_{(V, h)}(\*v) = \star$ for any $V$ and $h < n+1$: it could only be otherwise in the second case of \eqref{eq:sflowlevel} where since $\big|\mathbf{X}^h \setminus \mathbf{Y}^h\big| > 0$ there is here at least one $(X, h) < (V, h)$ with $\pi_{(X, h)}(\*v) = \star$ by the third case of \eqref{eq:sflowlevel}.
    Now if $h = n+1$ then we have $\pi_{(Y_1, n+1)}(\*v) = f_{V_{\varnothing}}(i)$ by the first case of \eqref{eq:sflowlevel} and (inductively) likewise for $Y_{2}, \dots, Y_{|\mathbf{Y}^{n+1}|}$. For any $V \notin \mathbf{Y}^{n+1}$, we have $\pi_{(V, n+1)}(\*v) = \star$.
    We thus obtain exactly the last two cases of \eqref{eq:lowlevelsolution}.
\end{proof}

To show that $p_\counterfactual^{\SCM_\LowLevel}$ marginalizes to $p_\counterfactual^{\RCM_\LowLevel}$, because $\SCM_\LowLevel$ and $\RCM_\LowLevel$ are constructed so as to share the same set of units and distribution over them, it suffices to show that for each $i \in \mathcal{U}$ we have that $(\SCM_\LowLevel)_\Ante, i \models (V, h) = f^\LowLevel_{(V, h)_\Ante}(i)$ for every $(V, h)_\Ante \in \RubinOutcomes_\LowLevel$.
This is easily seen by inspecting and comparing \eqref{eq:lowlevelsolution}, \eqref{eq:polowerlevel}, and \eqref{eq:prlowerlevel}.
\end{proof}

\subsection*{SUTVA}
The following remarks are in the context of the example from the discussion after Def.~\ref{def:constructiveabs} (see also that after Thm.~\ref{thm:abstractionrepresentation}). 

\begin{remark}\label{rmk:app:sutva}
We show how our framework can model unit-wise treatment assignment and the first part of SUTVA as making a particular abstraction viable.

Consider a model $\langle \{s\}, \{U, Y, T\}, \mathfrak{O}, \mathfrak{F}, \cdot\rangle$ where $s$ is some singleton element and $U$ is an endogenized unit variable (so that the exogenous unit set $\mathcal{U}$ is trivial in this model), $Y$ encodes survival with $\Val(Y) = \{0, 1\}$, and $T$ is a treatment variable that encodes whether \emph{each} unit was assigned to treatment, with $\Val(T) = \CartProd_{u \in \Val(U)} \Val(Y)$.
The list $\mathfrak{O}$ contains the $2^{|\Val(U)|} |\Val(U)|$ potential responses $Y_{\Ante}$ for each possible $\Ante$ where $\pi_T(\Ante) = \mathbf{u} \in \CartProd_{u \in \Val(U)} \{0, 1\}$ and $\pi_U(\Ante) = u \in \Val(U)$, and $\mathfrak{F}$ maps this to the survival outcome for patient $u$ when the vector $\mathbf{u}$ encodes treatment assignments for each patient.
Under the above SUTVA assumption, this outcome depends only on the component $\pi_u(\mathbf{u})$ of $\pi_T(\Ante)$ corresponding to $U = u$.
This means that a constructive abstraction $\tau$ to high-level variables $\{T', Y\}$ where $\Val(T') = \Val(U) \times \{0, 1\}$ with partition $\Pi$ is viable:
$\Pi_{T'} = \{T, U\}$, $\Pi_{Y} = \{Y\}$, and
$\tau(\Ante) = \big(\pi_{U}(\Ante), \pi_{\pi_U(\Ante)}(\pi_T(\Ante))\big)$.
In turn, given a distribution on units, we can exogenize them out of $T'$, showing that this model is equivalent to the familiar one with binary treatment and potential outcomes $Y_t(u)$.
\end{remark}

\begin{remark}\label{rmk:app:sutva2}
We show how to model several alternative ways mentioned by \cite{Imbens2015} of making the second part of SUTVA hold.

The first way
is for each unit to receive only one treatment level; then the possibility of inconsistency in abstraction is excluded.
Our framework requires that the potential responses are total functions (with domain the set of units)
and therefore that the outcome for each unit and treatment level is defined; nevertheless, it is possible to redo the analysis from the beginning up to causal abstraction, modifying Def.~\ref{def:rcm} to admit partial potential response functions. One would then find that this version of SUTVA guarantees the viability of the abstraction collapsing all treatment levels.

The second way is to admit stochastic treatments.
Suppose we have two treatment levels $\mathrm{tr}^i$, $\mathrm{tr}^j$ and some $u \in \mathcal{U}$, occurring with probability $P(u)$, for which $Y_{\mathrm{tr}^i}(u) \neq Y_{\mathrm{tr}^j}(u)$.
Then for any $p^i, p^j \ge 0$ with $p^i + p^j = 1$ we can form a new model with $\mathrm{tr}^i$, $\mathrm{tr}^j$ merged into a single treatment $\mathrm{tr}$, but $u$ split into two new units $u^i$, $u^j$ with respective probabilities $p^i P(u)$, $p^j P(u)$.
In this model we define $Y_{\mathrm{tr}}(u^i) = Y_{\mathrm{tr}^i}(u)$, $Y_{\mathrm{tr}}(u^j) = Y_{\mathrm{tr}^j}(u)$, and the violation to the condition in question is removed since we end up with only a single treatment level $\mathrm{tr}$. Further, the expectation of $Y_{\mathrm{tr}}$ conditioned on drawing $u^i$ or $u^j$ is equal to the weighted average of $Y_{\mathrm{tr}^i}(u)$ and $Y_{\mathrm{tr}^j}(u)$.
This is really a different definition in which we mix incompatible outcomes rather than consider them to invalidate the entire abstraction---study of such a notion may prove fruitful, but is outside the scope of this article.

A final way to make this condition hold is to coarsen the outcome, by, e.g. blurring multiple health statuses into the binary distinction between surviving and not.
This can in fact already be seen in our framework. Let $\tau$ be the abstraction from the running example coarsening treatment levels (values of $T$), $\tau'$ be that coarsening the outcome (values of $Y$), and $\mathcal{L}$ be our lowest-level model.
Then it is possible that we have some $\mathcal{H}', \mathcal{H}$ for which $\mathcal{H} \prec_\tau \mathcal{H}' \prec_{\tau'} \mathcal{L}$, but no $\mathcal{H}$ for which $\mathcal{H} \prec_\tau \mathcal{L}$. In this sense (by moving to the already coarser model $\mathcal{H}'$) the second part of SUTVA can be made to hold.
\end{remark}

\section{Inference (\S{}2)}\label{s:app:inference}

\begin{proposition}\label{prop:appx:late}
$\eqref{eq:pomonotonicity}, \eqref{eq:po:er}, \eqref{eq:po:od}\models \text{LATE} = \text{ITT}_1/\text{ITT}_2$.
\end{proposition}
\begin{proof}
Regarding the encoding of expected values: for any function $f: \Val_\counterfactual \to \mathbb{R}$ we can define an expected value $\mathbf{E}_{\mathbf{u}}(f) = \sum_{\mathbf{u} \in \mathcal{U}} P(\mathbf{u}) f(\mathfrak{O}(\mathbf{u}))$, abbreviated $\mathbf{E}(f)$.
If $f$ is integer-valued then $\mathbf{E}(f)$ is expressible within our formal language.
Thus LATE is equal to a ratio\footnote{The ratio $\mathbf{t}_1 = \nicefrac{\mathbf{t}_2}{\mathbf{t}_3}$ is an abbreviation for $\mathbf{t}_1 \cdot \mathbf{t}_3 = \mathbf{t}_2 \in \Ll$. Likewise in Ex.~\ref{ex:verma}.} of our terms,
\begin{align*}
\text{LATE} = \mathbf{E}(Y_{x^+} - Y_{x^-} | x^+_{z^+} \land x^-_{z^-}) = \frac{\mathbf{P}(y^+_{x^+} \land y^-_{x^-} \land x^+_{z^+} \land x^-_{z^-}) - \mathbf{P}(y^-_{x^+} \land y^+_{x^-} \land x^+_{z^+} \land x^-_{z^-})}{\mathbf{P}(x^+_{z^+} \land x^-_{z^-})}.
\end{align*}

We show $\text{LATE} = \text{ITT}_1/\text{ITT}_2$; \eqref{eq:pomonotonicity} shows the second term of $\text{ITT}_2$ vanishes, while the first term is the denominator of $\text{LATE}$, so it suffices to show equality of the numerator and $\text{ITT}_1$.
All of the following arithmetic moves rely only on \textsf{AX}. Note that
\begin{multline*}
     \mathbf{P}(y^+_{z^+, X_{z^+}} \land y^-_{z^-, X_{z^-}}) = \sum_{x, x'} \mathbf{P}\big(y^+_{z^+, x}\land x_{z^+} \land y^-_{z^-, x'} \land x'_{z^-} \big) = \sum_{x, x'} \mathbf{P}(y^+_{x} \land x_{z^+} \land y^-_{x'} \land x'_{z^-})
\end{multline*}
the last equality following from \eqref{eq:po:od} and \eqref{eq:po:er}.
Any terms for $x = x'$ vanish by \eqref{eq:po:er} (using also that $\mathbf{P}(\epsilon) = 0 \rightarrow \mathbf{P}(\epsilon \land \zeta) = 0$ is derivable from basic probabilistic logic), the term for $x = x^-$ and $x' = x^+$ vanishes by \eqref{eq:pomonotonicity}, and we are left with exactly the first term of the numerator.
We can analogously derive equality of the expressions with negative coefficients. 
\end{proof}

\begin{proposition}\label{prop:appx:itt1}
Under representability, $\text{ITT}_1 = \mathbf{E}(Y_{z^+} - Y_{z^-})$.
\end{proposition}
\begin{proof}
We claim that $\text{ITT}_1 = \mathbf{E}(Y_{z^+, X_{z^+}} - Y_{z^-, X_{z^-}}) = \mathbf{E}(Y_{z^+} - Y_{z^-})$ and thus that
\begin{align*}
\mathbf{P}\big(y^+_{z^+, X_{z^+}} \land y^-_{z^-, X_{z^-}}\big) - \mathbf{P}\big(y^-_{z^+, X_{z^+}} \land y^+_{z^-, X_{z^-}}\big)
= \mathbf{P}(y^+_{z^+} \land y^-_{z^-}) - \mathbf{P}(y^-_{z^+} \land y^+_{z^-})
\end{align*}
under representability.
It suffices to show $\text{Rep.} \models y_{z, X_z} \leftrightarrow y_z$ where $\text{Rep}.$ stands for representability.
Note that $\text{Rep.} \models y_{z, X_z} \leftrightarrow \bigvee_{x, y'} y_{zx} \land x_z \land y'_{z} \leftrightarrow \bigvee_{x} y_{xz} \land x_z \land y_z$, the last step since $\eqref{eq:pocomposition} \models x_z \land y'_z \rightarrow y'_{xz}$ contradicting $y_{xz}$ when $y' \neq y$. Similarly by setting $y' = y$, we have $\text{Rep}. \models \bigvee_{x} y_{xz} \land x_z \land y_z \leftrightarrow \bigvee_x x_z \land y_z \leftrightarrow y_z$. Note that no additional outcomes are used here except $Y_z$, which is necessary to state the result.
\end{proof}

\begin{proof}[Proof of Thm.~\ref{thm:poassumptionscompleteness}]
Soundness is straightforward.
For completeness, show that any consistent formula $\varphi$ is satisfiable.
By the completeness proof for $\textsf{AX}$ \citep[Thm.~6]{ibelingicard2020} we get a
distribution $P$ satisfying $\varphi$ over joint valuations $\Val(\mathfrak{O}_\varphi)$, where $\mathfrak{O}_\varphi$ is the set of potential outcomes appearing in $\varphi$. To construct a satisfying model $\RCM$, add a unit $u$ that witnesses any $\*o \in \Val(\mathfrak{O}_\varphi)$ for which $P(\*o) > 0$. It is clear that $\RCM \in \mathfrak{R}(S)$ since $\varphi$ is consistent with $\mathrm{T}(S)$.
\end{proof}

\begin{proof}[Proof of Cor.~\ref{cor:reffax}]
Observe that $\mathfrak{R}_{\mathrm{eff}} = \mathfrak{R}\eqref{eq:poeffectiveness}$.
\end{proof}

\begin{proof}[Proof of Thm.~\ref{thm:graphicalsystem}]
Soundness has been shown in prior works \cite{Pearl2009}, and in any case is straightforward ($\textsf{ER}^\Graph$ since each variable can only be a function of its parents in $\mathfrak{M}(\Graph)$ and $\textsf{cf-sep}^\Graph$ since the functional mechanisms determining unconfounded nodes must be independent; see also Cor.~\ref{cor:app:dsepcomplete} below for a derivation from d-separation).

Completeness: our proof is a modification of the proof for $\textsf{AX}_3$ in \cite[Thm.~6]{ibelingicard2020}.
We take the same steps as in that proof, the strategy being to show that any consistent $\varphi$ is satisfiable, up to and including the normal form in Lem.~8, except that we use the full set $\mathbf{V}$ in place of $\mathbf{V}_\varphi$ (this is possible since $\mathbf{V}$ is finite in our case).

Let $\mathcal{C}$ be the set of connected components under the edge relation $\Confounded$ of $\Graph$.
Each $\*C \in \mathcal{C}$ is a subset $\*C \subset \*V$ that is a complete graph under $\Confounded$, by assumption.
For each $\*C$ let 
\begin{align*}
\Delta_{\*C} = \Big\{ \bigwedge_{\substack{C \in \*C\\ \*p_C \in \Val(\mathbf{Pa}^\Graph_C)}} C_{\mathbf{p}_C} = c_{C, \*p_C} : \substack{c_{C, \*p_C} \in \Val(C) \text{ for each } C \in \*C,\\ \mathbf{p}_C \in \Val(\mathbf{Pa}_C^\Graph)} \Big\}
\end{align*}
be the collection of counterfactuals describing complete functional mechanisms for each $C \in \*C$.
The notation $C_{\mathbf{p}_C} = c_{C, \*p_C}$ represents the outcome that $C = c_{C, \*p_C}$ under the intervention $\mathbf{p}_C$; we use such notation in order to be explicit about the choice of variable $C$.
The below is our analogue of \cite[Lem.~8]{ibelingicard2020}.
\begin{lemma}
    \label{lem:normalform}
  Let $\varphi$ be a conjunction of literals.
  Then there are polynomial terms $\{\mathbf{t}_i, \mathbf{t}'_{i'}\}_{\substack{i, i'}}$
  in the variables $\big\{\mathbf{P}(\delta_{\*C}) : \substack{\delta_{\*C} \in \Delta_{\*C} \\ \*C \in \mathcal{C}}\big\}$
  such that \begin{align}\label{eqn:normalform}
\mathrm{T}(\Graph) \vdash \varphi\leftrightarrow \bigwedge_{\substack{\delta_{\*C} \in \Delta_{\*C}\\ \*C \in \mathcal{C}}} \mathbf{P}(\delta_{\*C}) \ge 0 \land \bigwedge_{\*C \in \mathcal{C}} \Big[\sum_{\delta_{\*C} \in \Delta_{\*C}} \mathbf{P}\left(\delta_{\*C}\right) = 1
\Big] \land \bigwedge_{i}
          \mathbf{t}_i \ge 0
         \land
        \bigwedge_{i'}
          \mathbf{t}'_{i'} > {0}. \end{align}
\end{lemma}
  \begin{proof}
    This will follow from the fact that we have the below \eqref{eq:deltafactorization} for any $\delta \in \Delta_{\mathrm{sat}}$.
    Here $\Delta_{\mathrm{sat}}$ is the set defined in \cite{ibelingicard2020} as the satisfiable subset of complete interventional state descriptions
    $\Delta = \Big\{ \bigwedge_{\substack{\Ante \in \Val(\*X) \\ \*X \subset \*V}} \bigwedge_{V \in \*V} V_\Ante = \pi_V(\*v_\Ante) : \substack{\*v_\Ante \in \Val(\*V)\\ \text{for each }\Ante} \Big\}$.
    \begin{align}\label{eq:deltafactorization}
    \mathrm{T}(\Graph) \Proves \mathbf{P}(\delta) =
    \prod_{\mathbf{C} \in \mathcal{C}}
    \mathbf{P}\Bigg(
      \bigwedge_{\substack{C \in \*C\\ \*p_C \in \Val(\mathbf{Pa}_C^\Graph)}} C_{\mathbf{p}_C} = c_{C, \mathbf{p}_C}
    \Bigg)
    \end{align}
    where $c_{C, \mathbf{p}_C} \in \Val(C)$ for each $C$ is the outcome of $C$ in $\delta$ in the outer conjunct corresponding to $\Ante = \mathbf{p}_C$ and the inner conjunct corresponding to $V = C$.
    To see \eqref{eq:deltafactorization}, note that for each $V \in \*V$ we can remove all conjuncts in a $\delta$ for any $\Ante$ that do not correspond to a setting ${\*p}_C$ of the parents $\mathbf{Pa}_C^\Graph$ by $\textsf{ER}^\Graph$ and composition, which we can use since $\textsf{SCM}$ is part of $\mathrm{T}(\Graph)$.
    We thereby obtain $\mathbf{P}(\delta) = \mathbf{P}\big(\bigwedge_{\substack{V \in \*V \\ \*p \in \mathbf{Pa}_V^\Graph}} V_{\*p} = \pi_V({\*v}_{\*p}) \big)$ and this gives us the $c_{V, \*p} = \pi_V({\*v}_{\*p})$.
    Next, we apply $\textsf{cf-sep}^\Graph$ to this expression, giving us the final product in \eqref{eq:deltafactorization} factorizing over all $\*C \in \mathcal{C}$.
    \end{proof}
    
    Now, as in the proof of \cite[Thm.~6]{ibelingicard2020}, the Positivstellensatz shows that consistency of $\varphi$ implies there is a real solution to \eqref{eqn:normalform} for $\big\{\mathbf{P}(\delta_{\*C})\big\}_{\substack{\delta_{\*C} \in \Delta_{\*C} \\ \*C \in \mathcal{C}}}$. We construct a model $\SCM$ with diagram $\Graph$ inducing these probabilities and thereby satisfying $\varphi$.
    Note that any $\delta_{\*C} \in \Delta_{\*C}$ defines a function $f^{\delta_{\*C}}_C: \Val(\mathbf{Pa}_C^\Graph) \to \Val(C)$ for each $C \in \*C$ via setting $f^{\delta_{\*C}}_C(\*p_C)$ to the value $C$ is set to under $\*p_C$ in $\delta_{\*C}$.
    We let $\SCM = \langle \{U_{\*C}\}_{\*C \in \mathcal{C}}, \*V, \{f_V\}_V, P \rangle$ where $P$ factors as $P\big(\{ u_{\*C}\}_{\*C \in \mathcal{C}} \big) = \prod_{\*C \in \mathcal{C}} P(u_\*C)$ and for any variable $C \in \*V$ with $C \in \*C$, we define parent sets $\mathbf{Pa}_C = \mathbf{Pa}_C^\Graph$ and $\mathbf{U}_C = \{U_{\*C}\}$; this implies that $\SCM$ has diagram $\Graph$ (there is a bidirected edge between any $C, C'$ in the same $\*C$, and the statement assumes $\Graph$'s bidirected arcs form a disjoint union of complete graphs).
    Finally let $\Val(U_{\*C}) = \Delta_{\*C}$ for each $\*C$, and let $P(\delta_{\*C})$ be the probability $\mathbf{P}(\delta_{\*C})$ from our solution to \eqref{eqn:normalform}; define the structural functions by $f_C(\*p_C, \delta_{\*C}) = f^{\delta_{\*C}}_C(\*p_C)$ for each $C \in \*C \in \mathcal{C}$. By construction, it is clear that $\SCM$ induces exactly the probabilities from our solution to \eqref{eqn:normalform}.
\end{proof}
The following definition of d-separation and its extension to mixed diagrams are standard (e.g., \cite{pearl1988probabilistic,ShpitserPearlComplete}). In a directed graph $\mathcal{D}$ over nodes $\mathbf{V}$, let $\mathbf{De}^{\mathcal{D}}_{V} = \{ V' \in \mathbf{V} : V \to \dots \to V'\}$ be the set of \emph{descendants} of $V \in \mathbf{V}$, the transitive closure of parenthood, with $V \in \mathbf{De}^{\mathcal{D}}_V$.
\begin{definition}[Directed d-separation]\label{def:app:dsepdir}
Given directed graph $\mathcal{D}$, the nodes $X$ and $Y$ are \emph{d-separated} given a set of nodes $\mathbf{Z}$ if every path (consisting of arrows of either direction) from $X$ to $Y$ either has a \emph{collider} node (to be defined shortly) on the path $M$ such that $\mathbf{De}_M^\mathcal{D} \cap \mathbf{Z} = \varnothing$, or a non-collider node $M'$ on the path such that $M' \in \mathbf{Z}$.
A collider is a vertex $M$ on the path such that the path is of the form $X \rightleftharpoons
 \dots \to M \leftarrow \dots \rightleftharpoons Y$ (the symbol $\rightleftharpoons$ denoting an arrow of either direction).
The subsets $\mathbf{X}$, $\mathbf{Y}$ of nodes are d-separated given $\mathbf{Z}$ if every $X \in \mathbf{X}$ and $Y \in \mathbf{Y}$ are d-separated given $\mathbf{Z}$.
\end{definition}
\begin{definition}[Mixed d-separation]\label{def:app:dsep}
Given mixed causal diagram $\Graph$ over $\mathbf{V}$, let $\mathcal{C}$ be the set of maximal cliques in $\Graph$ under its edge relation $\Confounded$. Thus every $\mathbf{C} \in \mathcal{C}$ is a subset $\mathbf{C} \subset \mathbf{V}$ of vertices such that $C \Confounded C'$ for every $C, C' \in \mathbf{C}$, and there is no $\mathbf{C}' \supsetneq \mathbf{C}$ with this property.
Then form the directed graph $\mathcal{D}^\Graph$ with nodes $\mathbf{V} \cup \mathcal{C}$ and edges
\begin{align*}
\{ V \to V' : V, V' \in \mathbf{V}, V \to V' \text{ in } \Graph \} \cup \{ \mathbf{C} \to V: \mathbf{C} \in \mathcal{C}, V \in \mathbf{V}, V \in \mathbf{C} \}.
\end{align*}
Finally, given $\mathbf{X}, \mathbf{Y}, \mathbf{Z} \subset \mathbf{V}$, we say that $\mathbf{X}$ and $\mathbf{Y}$ are d-separated in $\Graph$ given $\mathbf{Z}$ if the same holds in $\mathcal{D}^\Graph$ (according to Def.~\ref{def:app:dsepdir}).
\end{definition}
\begin{corollary}
\label{cor:app:dsepcomplete}
Let $\textsf{d-sep}^\Graph$ be an encoding of all the conditional independences implied by the graphical d-separation criterion over the diagram $\Graph$ (including over counterfactuals, where these are determined by parallel networks).
Then the system $\mathrm{T}(\textsf{ER}^\Graph) + \textsf{d-sep}^\Graph + \textsf{SCM}$ is sound and complete under the same assumptions as Thm.~\ref{thm:graphicalsystem}.
\end{corollary}
\begin{proof}
For soundness, see \cite{VermaPearl1988} or any standard discussion of d-separation, e.g., that in \cite{Pearl2009} or \cite{ShpitserPearlComplete}.
For completeness, in light of Thm.~\ref{thm:graphicalsystem} it suffices to show that any instance of $\textsf{cf-sep}^\Graph$, as in Def.~\ref{def:tgraph}, is an instance of $\textsf{d-sep}^\Graph$. Consider the parallel network $\Graph'$ with $n+n'$ copies of $\Graph$ and the two sets of nodes $\mathbf{Y} = \{(Y_i)_{\mathbf{p}_i}\}_{1 \le i \le n}$, $\mathbf{Y}' = \{(Y'_i)_{\mathbf{p}'_i}\}_{1 \le i \le n'}$. Since there is no $Y \in \mathbf{Y}$, $Y' \in \mathbf{Y}'$ for which $Y = Y'$ or $Y \Confounded Y'$ in $\Graph$ and all ingoing edges to $Y$ and $Y'$ have been severed in $\Graph'$, the only possibility for a (direction-agnostic) path between some such $Y, Y'$ is a bidirected path of length $\ge 2$ or a bidirected path between proper descendants of $Y, Y'$.
All such paths contain colliders, so that d-separation delivers the (unconditional) independence \eqref{eqn:cfsep}.
\end{proof}

\end{document}